\documentclass[oneside,english]{amsart}
\usepackage[T1]{fontenc}
\usepackage[latin9]{inputenc}
\usepackage{geometry}
\usepackage{amsthm}
\usepackage{amstext}
\usepackage{amssymb}
\usepackage{graphicx}
\usepackage{setspace}
\usepackage{esint}
\usepackage{mathrsfs}
\usepackage{enumerate}

\newcommand{\sL}{\mathcal L}
\newcommand{\sM}{\mathcal M}

\newcommand{\E}{\mathbb E}
\newcommand{\R}{\mathbb R}
\newcommand{\sE}{\mathcal E}
\newcommand{\sG}{\mathcal G}
\newcommand{\sH}{\mathcal H}
\newcommand{\sV}{\mathcal V}
\newcommand{\sgn}{\operatorname{sgn}}

\numberwithin{equation}{section}
\numberwithin{figure}{section}
\theoremstyle{plain}
\newtheorem{thm}{\protect\theoremname}
  \theoremstyle{remark}
  \newtheorem*{rem*}{\protect\remarkname}
  \theoremstyle{plain}
  \newtheorem{lem}[thm]{\protect\lemmaname}
  \theoremstyle{plain}
  \newtheorem{prop}[thm]{\protect\propositionname}
  \theoremstyle{plain}
  \newtheorem{cor}[thm]{\protect\corollaryname}
   \theoremstyle{plain}
  \newtheorem*{thm*}{\protect\theoremname}
  
  \theoremstyle{plain}
    \theoremstyle{remark}
  \newtheorem{rem}[thm]{\protect\remarkname}
  \theoremstyle{plain}
  \newtheorem{ass}{Standing Assumption}

\makeatother

  \providecommand{\corollaryname}{Corollary}

  \providecommand{\lemmaname}{Lemma}
    \providecommand{\propositionname}{Proposition}
  \providecommand{\remarkname}{Remark}
  \providecommand{\theoremname}{Theorem}
\providecommand{\theoremname}{Theorem}
\providecommand{\resultname}{Result}

\author{David Hobson \and Yeqi Zhu}
\thanks{Department of Statistics, University of Warwick, Coventry, CV4 7AL, UK. D.Hobson@warwick.ac.uk, Yeqi.Zhu@warwick.ac.uk}

\usepackage[english]{babel}

\begin{document}

\title{\textbf{Multi-asset consumption-investment problems with infinite transaction costs}}

\date{\today}

\begin{abstract}
The subject of this paper is an optimal consumption/optimal portfolio problem with transaction costs and with multiple risky assets.

In our model the transaction costs take a special form in that transaction costs on purchases of one of the risky assets (the endowed asset) are infinite, and transaction costs involving the other risky assets are zero. Effectively, the endowed asset can only be sold. In general, multi-asset optional consumption/optimal portfolio problems are very challenging, but the extra structure we introduce allows us to make significant progress towards an analytical solution.

For an agent with CRRA utility we completely characterise the different types of optimal behaviours. These include always selling the entire holdings of the endowed asset immediately, selling the endowed asset whenever the ratio of the value of the holdings of the endowed asset to other wealth gets above a critical ratio, and selling the endowed asset only when other wealth is zero. This characterisation is in terms of solutions of a boundary crossing problem for a first order ODE. The technical contribution is to show that the problem of solving the HJB equation, which is a second order, non-linear PDE subject to smooth fit at an unknown free boundary, can be reduced to this much simpler problem involving an explicit first order ODE. This technical contribution is at the heart of our analytical and numerical results, and allows us to prove monotonicity of the critical exercise threshold and the certainty equivalent value in the model parameters.



\smallskip
\noindent \textbf{Keywords and phrases:}  optimal consumption/investment problem, transaction costs, multiple correlated assets, singular stochastic control, reflecting diffusion, Skorokhod problem

\smallskip
\noindent \textbf{MSC 2010 subject classifications:} Primary  93E20; Secondary
35R35, 49J15, 49L20, 60J55

\noindent \textbf{JEL classifications:} C61; D23; D52; G11

\end{abstract}

\maketitle

\onehalfspacing

\section{Introduction}

In one of his seminal works, Merton~\cite{Merton} considered the portfolio and consumption problem faced by a price-taking agent in a continuous-time stochastic model consisting of a risk-free bond and a risky asset. The agent is assumed to have the objective of maximising discounted expected utility from consumption
over the infinite horizon. In a model in which the single risky asset follows an exponential Brownian motion with constant parameters, and the agent has constant relative risk aversion, Merton showed that optimal behaviour is to consume at a rate which is proportional to wealth, and to invest a constant fraction of wealth in the risky asset.

Constantinides and Magill~\cite{ConMagill} were the first to add proportional transaction costs to the model. They conjectured the form of the optimal strategy, namely that is is optimal to keep the fraction of wealth invested in the risky asset in an interval. Subsequently Davis and Norman~\cite{Davis} gave a precise statement of the result and showed how the solution could be expressed in terms of local times. Optimal behaviour is to trade in a minimal fashion so as to keep the variables (cash wealth, wealth in the risky asset) in a wedge-shaped region in the plane, and this is achieved by sales and purchases of the risky asset
in the form of singular stochastic controls.

The approach in Davis and Norman~\cite{Davis} is to write down the Hamilton-Jacobi-Bellman equation, and to characterise the candidate value function as a solution to this equation. Subsequently, Shreve and Soner~\cite{SS} reproved many of the results of \cite{Davis} using viscosity solutions. These approaches remain the main methods for solving portfolio optimisation problems with transaction costs, although recently a different technique based on shadow prices has been proposed, see Guasoni and Muhle-Karbe~\cite{Guasoni} for a users' guide.

The results in Davis and Norman~\cite{Davis} are limited to a single risky asset, and it is of great interest to understand how they generalise to multiple risky assets.
In his survey article on consumption/investment problems with transaction costs Cadenillas~\cite[page 65]{Cadenillas} says that `most results in this survey are limited to the case of only one bond and only one stock. It is then important to see if these results can be extended to cover a realistic number of stocks'. Although there has been some progress since that paper was published, similar sentiments are echoed in recent papers by Chen and Dai~\cite[page 2]{ChenDai}: `most of the existing theoretical characterisations of the optimal strategy are for the single risky-asset case. In contrast there is a relatively limited literature on the multiple risky-asset case' and Guasoni and Muhle-Karbe~\cite[page 194]{Guasoni}: `In sharp contrast to frictionless models, passing from one to several risky asstes is far from trivial with transaction costs \ldots multiple assets introduce novel effects, which defy the one-dimensional intuition'. In summary therefore, there is great interest in both theoretical and numerical results on the multi-asset case, and this paper can be considered as a contribution to that literature.

In the multi-asset case, and on the computational side, Muthuraman and Kumar~\cite{MuthKumar} use a process of policy improvement to construct a numerical solution for the value function and the associated no-transaction region, and Collings and Haussman~\cite{Collings} derive a numerical solution via a Markov chain approximation, for which they prove convergence. On the theoretical front Akian {\em et al}~\cite{Akian} show that the value function is the unique viscosity solution of the HJB equation (and provide some numerical results in the two-asset case) and Chen and Dai~\cite{ChenDai}
identify the shape of the no-transaction region in the two-asset case. Explicit solutions of the general problem remain very rare.

One situation when an explicit solution is possible is the rather special case of uncorrelated risky assets, and an agent with constant absolute risk aversion. In that case the problem decouples into a family of optimisation problems, one for each risky asset, see Liu~\cite{liu1}.
Another setting for which some progress has been made is the problem with small transaction costs, see Whalley and Willmott~\cite{Whalley}, and for a more recent analysis Soner and Touzi~\cite{Soner}. Whalley and Willmott use an expansion method to provide asymptotic formulae for the optimal strategy.

In this paper we consider a different extreme in which transaction costs are either zero or (with regard to trades in one direction of a particular risky asset) infinite. Expressed differently, one of the risky assets is assumed to not be available for dynamic trading (either to the market as a whole, or to the utility-maximising agent, perhaps for legal reasons, or simply because it is difficult for individuals to trade particular stocks actively). Instead the assumption is that this asset can only be sold: (re)-purchases are not allowed. Our agent is endowed with an initial quantity of this asset, and her strategies include when to sell units of this (infinitely divisible) asset over time. The assumption that there are infinite transaction costs on purchases of endowed asset is complemented by an assumption that sales and purchases of the other risky assets are permissible, and incur zero transaction costs.

The set-up of our problem in which one asset is identified as a non-traded asset is similar to that in the real options literature (Miao and Wang~\cite{Miao}, Henderson~\cite{Henderson}, Henderson and Hobson~\cite{hh}) in which an agent with the option to invest in a project (or sell an asset) chooses the optimal sale time. The difference with respect to that literature is that we assume that the non-traded asset is infinitely divisible, whereas in the real options literature it is typically assumed to be indivisible.

Our model consists of an agent who is endowed with units of an infinitely divisible risky asset
which may be sold, but not bought, and whose opportunity set includes investment in a risk-free bond, and investment in other risky assets to which a zero transaction cost applies. The risky assets follow correlated exponential Brownian motions and the objective of the agent is to choose a consumption rate and an investment strategy (including a sale strategy for the endowed asset) so as to maximise the discounted expected utility of consumption.

This paper is an extension of Hobson and Zhu~\cite{HobsonZhu14a} which considers a similar problem with an endowed asset but with no other risky assets. Many of the techniques of \cite{HobsonZhu14a} carry over to the wider setting of this paper, however, since there are fewer parameters when there are no investments beyond the endowed asset, the problem in \cite{HobsonZhu14a} is significantly simpler and much more amenable to a comparative statics analysis. In contrast, this paper treats the multi-asset problem which has proved so difficult to analyse in full generality, albeit in a rather special case. The multi-asset setting brings new challenges and complicates the analysis.

It is straightforward to write down the Hamilton-Jacobi-Bellman (HJB) equation for our problem. In general the value function is a function of four variables (wealth in liquid assets, price of the endowed asset, quantity of endowed asset held, time) and satisfies a HJB equation which is second order, non-linear and subject to value matching and smooth fit at an unknown free-boundary. (The smooth fit turns out to be of second order.) In our special setting the problem of finding the free-boundary and value function is reduced to study of a boundary crossing problem for the solution of a first-order ordinary differential equation (ODE). This first crossing problem has four types of solution (`crossing' at zero, crossing in $(0,1)$, no crossing in $[0,1]$ and hits zero before crossing). Each different type of solution is identified with a different type of solution to the optimisation problem; for example the first type of solution corresponds to a strategy of immediately selling all units of the endowed asset. It is relatively straightforward to identify the parameter combinations which lead to different types of solution, even if explicit solutions of the first order ODE are not available. Then we can relate the optimal wealth process, consumption strategy, sale strategy and investment strategy of the agent path-wise to the solution of a Skorokhod-type problem.

\section{The model and main results}

\subsection{Problem Formulation}

Let $(\Omega,\mathcal{F},\mathbb{P}, {\mathbb F} = \left(\mathcal{F}_{t}\right)_{t\geq0})$ be  a filtered probability
space, such that the filtration satisfies the usual conditions and is generated
by a two-dimensional\footnote{We can consider the case of $n$ Brownian motions, and $n-1$ financial assets and a single endowed asset, but the situation reduces to the case considered here. In particular, the $(n-1)$ financial assets reduce to a single mutual fund.} standard Brownian motion $\left(B_{t}^{1},B_{t}^{\perp}\right)_{t\geq0}$. Set $B_{t}^{2} = \rho B_{t}^{1} + \rho^{\perp} B_{t}^{\perp}$, where $(\rho^{\perp})^2 = 1 - \rho^2$.

The financial market is modelled with three stochastic processes on this space, a bond paying a constant rate of interest $r$, a financial (or hedging) asset
with price process $P = (P_t)_{t \geq 0}$ and a non-traded (or endowed) asset with price process $Y = (Y_t)_{t \geq 0}$.
Assume the price processes of the risky assets satisfy
\[
P_{t}=p_0 \exp\left\{ \left(\mu-\frac{\sigma^{2}}{2}\right)t+\sigma B_{t}^{1}\right\} ,
\]
where $\mu$ and $\sigma>0$ are the constant mean return
and volatility of the financial asset, and $p_{0}$ is the initial price, and
\[
Y_{t}=y_{0}\exp\left\{ \left(\alpha-\frac{\eta^{2}}{2}\right)t+\eta B_{t}^{2}\right\} ,
\]
where $\alpha$ and $\eta>0$ are the constant mean return and volatility
of the non-traded asset, and $y_{0}$ is the initial price.
Let $\lambda = (\mu - r)/\sigma$ and let $\zeta = (\alpha - r)/\eta$ be the Sharpe ratios of the hedging and endowed assets respectively.

Let $C=\left(C_{t}\right)_{t\geq0}$ denote the consumption rate of
the individual, let $\Theta=\left(\Theta_{t}\right)_{t \geq 0}$
denote the number of units of the endowed asset held by the
investor and let $\Pi = (\Pi_t)_{t \geq 0}$ denote the cash amount
invested in the hedging asset $P$.
The consumption rate is required to be progressively measureable and
non-negative, the process $\Theta$ is required to be progressively measureable,
right-continuous with left limits and non-increasing to
reflect the fact that the non-traded asset is only allowed for
sale and $\Pi$ is required to be
progressively measurable.
We assume the initial number of shares held by the investor
is $\theta_{0}$. Since we allow for an initial
transaction at time 0 we may have $\Theta_0 < \theta_0$. We write
$\Theta_{0-} = \theta_0$. This is consistent with our convention
that $\Theta$ is right-continuous.

We denote by $X=\left(X_{t}\right)_{t\geq0}$ the wealth process of
the individual, and suppose that the initial wealth is $x_{0}$. Provided
the changes to wealth occur from either consumption, investment or from the
sale of the endowed asset, $X$ evolves according to
\[ 
dX_{t}=\Pi_{t}\frac{dP_{t}}{P_{t}}+\left(X_{t}-\Pi_{t}\right)r dt -C_{t}dt-Y_{t}d\Theta_{t}
= \sigma \Pi_t dB^1_t + \left\{ (\mu - r) \Pi_t + rX_t - C_t \right\} dt - Y_t d \Theta_t, 
\] 
subject to $X_{0-} = x_0$, and $X_0 = x_0 + y_0(\theta_0 -
\Theta_0)$.
We say a consumption/investment/sale strategy triple is admissible if the
components satisfy the requirements listed above and if the
resulting cash wealth process $X$ is non-negative for all time. Let
$\mathcal{A}\left(x,y,\theta\right)$ denote the set
of
admissible strategies for initial setup
$\left(X_{0-}=x,Y_{0}=y,\Theta_{0-}=\theta\right)$.

The objective of the agent is to maximise over admissible strategies
the discounted expected utility of consumption over the infinite horizon,
where the utility function of the agent is assumed
to have constant relative risk aversion,\footnote{The techniques extend to the case $R=1$ and logarithmic utility, but we will not consider that case here. However, many of the results can be obtained simply by setting $R=1$ in the various formulae.} with parameter $R \in (0,\infty) \setminus {1}$ and discount factor $\beta$.
In particular, the goal is to find ${\mathcal V} = \sV(x_0,y_0,\theta_0)$ where
\begin{equation}
{\mathcal V}(x,y,\theta) = \underset{\left(C,\Pi,\Theta\right)\in\mathcal{A}(x,y,\theta)}
{\sup}\mathbb{E}
\left[\int_{0}^{\infty}e^{-\beta t}\frac{C_{t}^{1-R}}{1-R}dt\right].
\label{eq:100}
\end{equation}
Since the set-up has a Markovian structure, we expect the optimal consumption, optimal investment and optimal sale strategy to be of
feedback form and to be functions of the current wealth and endowment of the agent and of the price of the risky asset.

Let $V(x,y,\theta,t)$ be the forward starting value function
for the problem
so that
\begin{equation} V(x,y,\theta,t) =
\underset{\left(C,\Pi,\Theta\right)\in\mathcal{A}(x,y,\theta,t)}
{\sup}\mathbb{E}\left[ \left. \int_{t}^{\infty}e^{-\beta
s}\frac{C_{s}^{1-R}}{1-R}ds \right| X_{t-} = x, Y_t=y, \Theta_{t-} =
\theta\right].
\label{eq:100b}
\end{equation}
Here the space of admissible strategies
$\mathcal{A}(x,y,\theta,t)$ is such that $C = (C_s)_{s \geq t}$ is
a non-negative, progressively measureable process, $\Pi = (\Pi_s)_{s \geq t}$ is
progressively measurable, $\Theta = (\Theta_s)_{s \geq t}$ is
a right-continuous, non-increasing, progressively measureable process and satisfies
$\Theta_t - (\Delta \Theta)_t = \theta$, and $X = (X_s)_{s \geq t}$ is a non-negative process given by
\[
X_s = x + \int_t^s \left(\lambda\sigma\Pi_{u} + rX_{u} - C_{u}\right) du
+ \int_t^s \sigma\Pi_{u}dB_{u}^{1} - \int_{[t,s]} Y_u d \Theta_u.
\]

It is clear that $V(x,y,\theta,t) = e^{-\beta t}V(x,y,\theta,0) = e^{-\beta t} {\mathcal V}(x,y,\theta)$.
Define $V_0(x,t) = V(x,y,0,t)$ (note $V_0$ will not depend on $y$, since the agent has no units of $Y$) and ${\mathcal V}_0(x) = {\mathcal V}(x,y,0)$, and
define the certainty equivalent price $p=p(x,y,\theta,t)$ as the solution to
\[ V_0(x+p,t) = V(x,y,\theta,t) .\]
Note that $p=p(x,y,\theta)$ does not depend on $t$ since $p$ solves ${\mathcal V}_0(x+p) = {\mathcal V}(x,y,\theta)$.

The parameters of the problem are $r$, $\beta$, $\mu$, $\sigma$, $\alpha$, $\eta$, $\rho$ and $R$ which we assume to all be constants. We will assume that $\rho \in (-1,1)$
(the limiting cases $\rho=\pm 1$ can be dealt with by similar techniques).  Define auxiliary parameters $(b_{i})_{1\leq i\leq4}$,
\[
b_{1} = \frac{2}{\eta^2 (1 - \rho^2)} \left[\beta - r(1-R) - \frac{\lambda^2 (1-R)}{2R}\right],
\hspace{4mm}
b_{2}= \frac{\lambda^2 - 2 R \eta \rho \lambda + \eta^2 R^2}{\eta^2 R^2 (1 - \rho^2)},
\hspace{4mm}
b_{3}= \frac{2(\zeta - \lambda \rho)}{\eta (1 - \rho^2)},
\]
and $b_4 = \left(\frac{1}{2} \eta^2 (1 - \rho^2)\right)^{-1}$.
Note that
\[
b_{2} = 1 + \frac{1}{1 - \rho^2} \left(\frac{\lambda}{\eta R} - \rho\right)^2 \geq 1.
\]
It will turn out that the optimal selling and investment problem depends on the original parameters only through these
auxiliary parameters and the risk aversion $R$.

We will see in later sections that $b_1$ plays the role of a `normalised discount factor'.
The parameter $b_3$ is the `effective Sharpe ratio, per unit of idiosyncratic volatility' of the endowed asset.
The parameter $b_2$ is the hardest to interpret: essentially it is a
nonlinearity factor which arises from the multi-dimensional structure of the problem.
The case $b_2 = 1$ is rather special and will be excluded to a certain extent from our
analysis.
(One scenario in which we naturally find $b_2 = 1$ is if $\lambda = 0 = \rho$. In this case there is neither a hedging motive, nor an investment motive
for holding the financial asset. Essentially then, the investor can ignore the presence of the financial asset, reducing the dimensionality of the problem.
This is the problem considered in Hobson and Zhu~\cite{HobsonZhu14a}. That paper uses many of the same ideas as this paper,
albeit in a much simpler setting and can be seen as `warm-up' problem for the more general problem considered here.)

Our goal is to solve for $V$, and hence for the certainty equivalent price $p$. As might be expected, $V$ solves a variational principle, and can be
characterised by a second-order, nonlinear partial differential equation in the four variables $(x,y,\theta,t)$ subject to value matching and
smooth fit (of the first and second derivatives)
at an unknown free boundary. In fact various simplifications can be expected from the inherent scalings of the problem. Nonetheless, the remarkable fact on
which this paper is based is that expressions for $V$ and a characterisation for the optimal solution all follow from the study of a boundary crossing problem for a single first order ordinary differential equation.

If $R<1$ and $b_1 \leq 0$ then the value function $V_0(x,t)$ is infinite for the Merton problem (in the absence of the endowed asset), and {\em a fortiori} the value function
with a positive endowment of the non-traded asset is also infinite. In this case it is not possible to define a certainty equivalent price. If $R>1$ and
$b_1 \leq 0$, then for every admissible strategy with zero initial endowment of  the risky asset the expected discounted utility of consumption equals $-\infty$, and again it is not possible to define a certainty equivalent price for units of the endowed asset. To exclude these cases we make the following
non-degeneracy assumption:

\begin{ass}
Throughout the paper we assume that $b_1>0$.
\end{ass}

\subsection{Main results}
The key to our analysis are solutions to the first order differential equation (\ref{eq:n}) the properties of which are stated in Proposition~\ref{prop:main},
the proof of which is given in Appendix~\ref{app:n}.

\begin{figure}
\begin{center}
\includegraphics[scale=0.3]{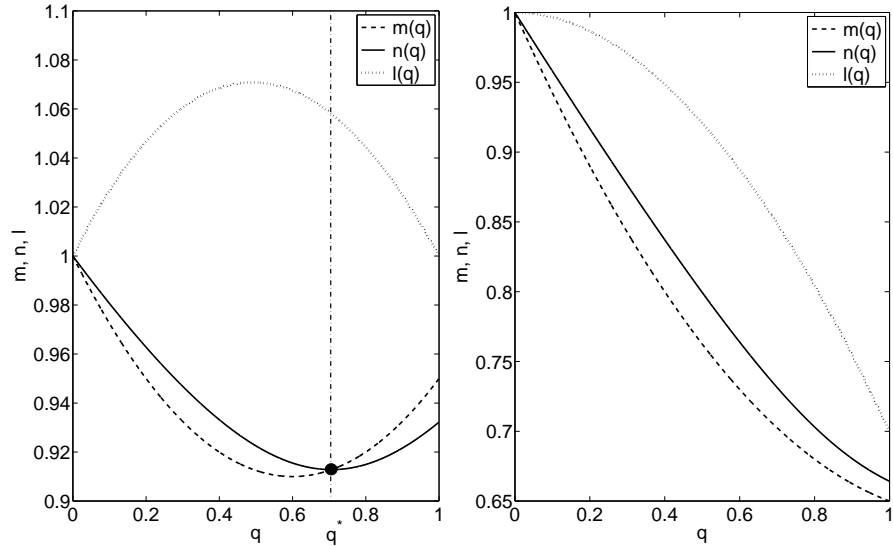}
\end{center}

\caption{Stylised plot of $m(q)$, $n(q)$, $\ell(q)$ for $R \in (0,1)$.
Parameters are such that $q^* \in (0,1)$ (left figure) and $q^* = 1$ (right figure).}
\label{fig:mnlRlessthan1}
\end{figure}

\begin{figure}
\begin{center}
\includegraphics[scale=0.3]{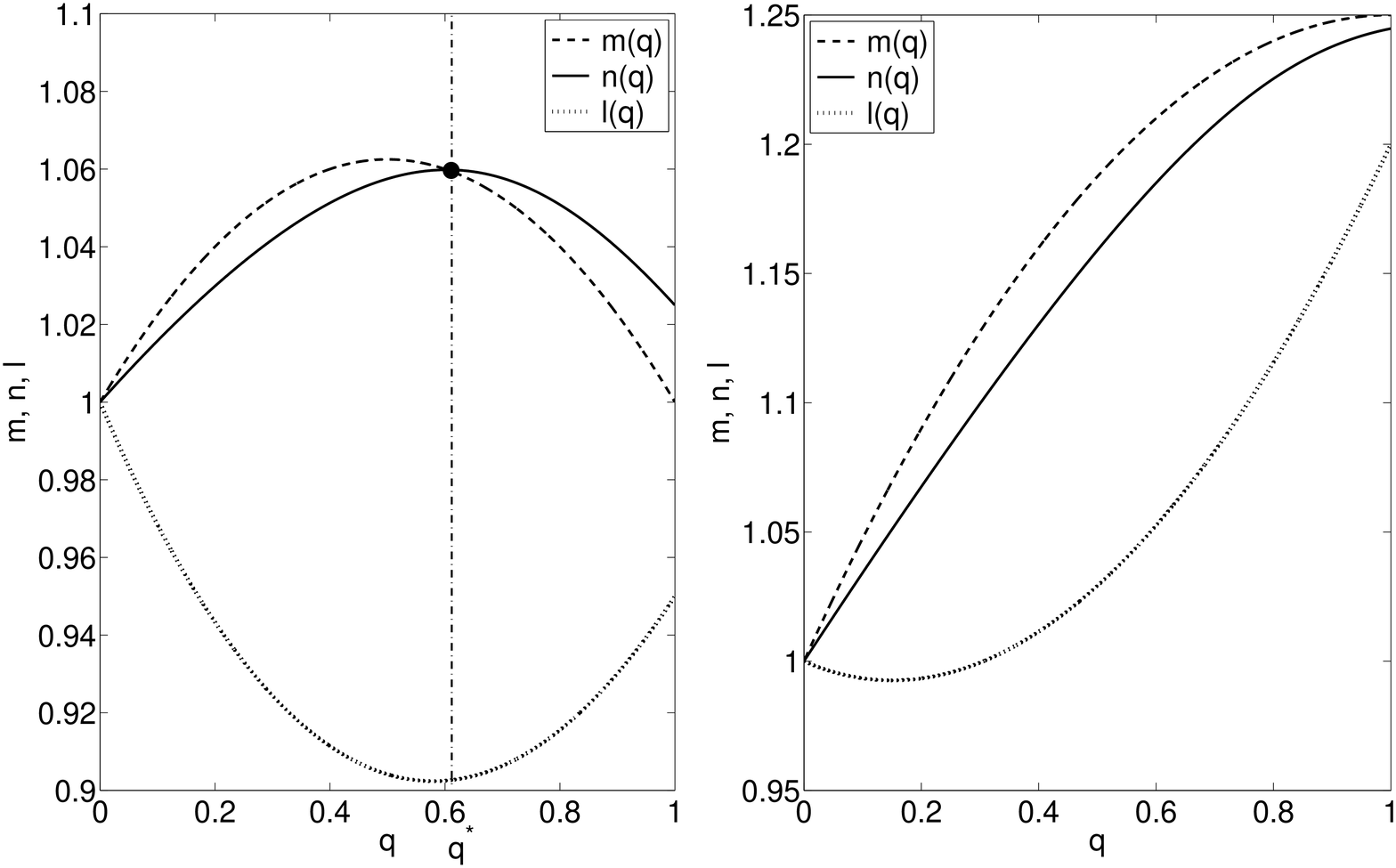}
\end{center}

\caption{Stylised plot of $m(q)$, $n(q)$, $\ell(q)$ for $R \in (1,\infty)$.
Parameters are such that $q^* \in (0,1)$ (left figure) and $q^* = 1$ (right figure).}
\label{fig:mnlRbiggerthan1}
\end{figure}

\begin{prop}
\label{prop:main}
For $q \in [0,1]$ define $m(q) = \frac{(1-R)R}{b_1} q^2 - \frac{b_3(1-R)}{b_1} q + 1$
and
$\ell(q) = m(q) + \frac{1-R}{b_1}q(1-q) + \frac{(b_2 - 1)R(1-R)}{b_1} \frac{q}{[(1-R)q + R]}.$
Let
$n = n(q)$ solve
\begin{equation}
\frac{n^{'}(q)}{n(q)} = \frac{1-R}{R(1-q)} - \frac{(1-R)^2}{b_1 R}\frac{q}{\ell(q) - n(q)} +
\frac{(1 - R)q}{2b_1 R (1-q) [(1-R)q + R]} \frac{\upsilon(q, n(q))}{\ell(q) - n(q)}
\label{eq:n}
\end{equation}
subject to $n(0) = 1$ and $\frac{n'(0)}{1-R} < \frac{\ell'(0)}{1-R} = (b_2 - b_3)/b_1$,
where
\[\upsilon(q, n) = \varphi(q,n) - \sgn(1-R) \sqrt{\varphi(q, n)^2  + 4R^2 (1-R)^2 (b_2 - 1) (1-q)^2},
\]
and $\varphi(q, n) = b_1 n + (1 - R)(b_3 - 2R) q + 2R(1-R) - b_1
- b_2R(1-R)$.

Suppose that if $n$ hits zero, then $0$ is absorbing for $n$.

For $R<1$, let $q^* = \inf\{q>0:n(q) \leq m(q)\}$, see {\rm Figure~\ref{fig:mnlRlessthan1}}. For $R>1$, let $q^* = \inf\{q>0: n(q) \geq m(q)\}$, see {\rm Figure~\ref{fig:mnlRbiggerthan1}}.
For $j \in \{\ell, m, n\}$ let $q_j = \inf\{q>0: j(q) = 0\} \wedge 1$.

Set $\bar{b}_3 = 2R$ if $R>1$ and $\bar{b}_3 = \min \{2R, R + \frac{b_1}{1-R} \}$ if $R<1$.
For fixed $b_1$, $b_2$ and $R$, there exists
some critical value $b_{3, crit} (b_1, b_2, R)$, with $R <  b_{3,crit} \leq \bar{b}_3 $, and such that
\begin{enumerate}
\item if $b_3 \leq 0$, then $q^* = 0$;
\item if $0 < b_{3} < b_{3,crit}$ then $0 < q^* < 1$;
\item if $R>1$ and $b_3 \geq b_{3,crit} (b_1, b_2, R)$ then $q^*=1$; if  $R<1$ and $b_{3,crit} (b_1, b_2, R) \leq b_3 < \frac{b_1}{1-R} + b_2 R$, then $q^* = 1$;
\item if $R<1$, $b_2=1$ and $b_3 = \frac{b_1}{1-R} + R \geq 2R$, then $q_m = q_n = q_{\ell} = q^*=1$;
if $R<1$, $b_2=1$ and $b_3 = \frac{b_1}{1-R} + R < 2R$, then $q_m < q_n = q_{\ell} = q^*=1$;
if $R<1$, $b_2>1$ and $b_3 = \frac{b_1}{1-R} + b_2 R$, then $q_m < q_n = q_{\ell} = q^*=1$;
if $R<1$  and $b_3 > \frac{b_1}{1-R} + b_2 R$, then $q_m < q_n = q_{\ell} < q^* = 1$.
\end{enumerate}

\end{prop}

\begin{rem}

The condition $b_3<2R$ is equivalent to $m'(1) > 0$.

If $R<1$,
then the condition $b_3 \leq \frac{b_1}{1-R} + b_2 R$ is equivalent to $\ell(1) \geq 0$. (Note that $\ell(1) \geq 0$ is a necessary condition for $q_n = 1$.)
Then, if $R<1$, $0<b_3<2R$ and $b_3 < \frac{b_1}{1-R} + b_2 R$, we have $q_{\ell} = q_{n} = 1$.

We will show in Lemma~\ref{lem:n} that $n$ has a turning point at $q^*\in(0,1)$ if and only if $n(q^*) = m(q^*)$. In particular, if $m$ is monotone, then $q^* = 1$.

\end{rem}

\begin{rem}
Suppose $b_2=1$. Then $\ell(1)=m(1)$. Moreover $\varphi(q,m) = R(1-R)(1-q)^2$. If $R<1$, then if $n \geq m$
we have $\varphi(q,n) \geq \varphi(q,m)>0$ and $\upsilon(q,n)=0$. Conversely, if $R>1$ then if $n \leq m$,
$\varphi(q,n) \leq \varphi(q, m) < 0$
and $\upsilon(q,n)=0$. Hence the expression in (\ref{eq:n}) for $n'$ simplifies greatly if $b_2=1$, and is seen to reduce to the
equation for the variable of the same name in \cite[Equation (3.6)]{HobsonZhu14a}.
\end{rem}

\begin{rem}
We show below in Lemma~\ref{lem:monbi} that $n(q)$, $q^*$ and $n(q^*)$ are each monotonic in the parameter $b_1$, $b_2$ and $b_3$ for $q \leq q^*$. In particular, $q^*$ is an increasing function of $b_3$.
It follows that there exists a critical parameter $b_{3,crit}$ and $q^*<1$ if and only if
$b_3 < b_{3,crit}$. Although we can conclude that $R<b_{3,crit} \leq \min \{2R, \frac{b_1}{1-R} + R \}$, we do not have an explicit expression for $b_{3,crit}$.
\end{rem}

From the scalings of the problem,  it is clear that a key variable is the ratio of wealth in the endowed asset to liquid wealth. (Here we define liquid wealth to be the sum of cash wealth and wealth invested in the hedging asset.) We denote this ratio
by $Z$ so that $Z_t= Y_t \Theta_t/X_t \in [0,\infty]$. Under optimal behaviour, consumption and investment rates are functions of liquid  wealth and $Z$.

One of the key contributions of this article is to identify the different types of solutions to the optimisation problem with different classes of solutions to the first
crossing problem studied in Proposition~\ref{prop:main}.


\begin{thm}
\label{thm:4cases}

\begin{enumerate}
\item Suppose $b_3 \leq 0$. Then it is always optimal to sell the
entire holding of the endowed asset immediately, so that
$\Theta_{t}=0$ for $t\geq0$. The value function for the problem is
$V\left(x,y,\theta,t\right) = \left(\frac{b_1}{b_4 R}\right)^{-R} e^{-\beta t} (x + y\theta)^{1-R}/1-R$; and
the certainty equivalent value of the holdings of the asset is
$p(x_0,y_0,\theta_0,0)=y_0
\theta_0$.

\item Suppose $0 < b_3 < b_{3,crit}(b_1, b_2, R)$. Then there exists a positive and finite
critical ratio $z^*$ and the optimal behaviour is to sell sufficient
units of the risky asset so as to keep the ratio of wealth in the risky
asset to cash wealth below the critical ratio. If $\theta>0$ then
$p(x,y,\theta,t) > y \theta$.

\item Suppose $b_3 \geq b_{3,crit}(b_1, b_2, R)$ and
$b_3 < \frac{b_1}{1-R} + b_2 R$ if $R<1$. Then the critical ratio $z^*$ is infinite
and the optimal behaviour is first to consume liquid wealth and invest in the risky asset, and then when this liquid wealth is exhausted, to finance further consumption and investment in the risky asset from sales of the endowed asset.

\item Suppose $b_3 \geq \frac{b_1}{1-R} + b_2 R$ if $R<1$.
Then the problem is degenerate, and provided $\theta_0$ is positive, the
value function $V=V(x,y,\theta,t)$ is infinite. There is no unique
optimal strategy, and the certainty equivalent value $p$ is not defined.
\end{enumerate}
\end{thm}

\begin{rem}

This theorem emphasises the role played by the parameter $b_3$, the `effective Sharpe ratio' to distinguish between the different scenarios.
When $b_3$ is negative, the endowed asset is a bad investment and it is optimal to sell it immediately.
For small and positive $b_3$, there exists a finite critical ratio and sales of the nontraded asset occur to keep the fraction of wealth held in the nontraded asset
below a critical value. As $b_3$ becomes larger, the endowed asset is more valuable and the agent waits longer for a better return from the endowed asset.
For sufficiently large $b_3$ she does not make any sales of the endowed asset until cash wealth is exhausted. Finally, if $R<1$ and $b_3$ becomes too large,
the value function is infinite, and the problem with the endowment is ill-posed.

\end{rem}

\begin{figure}
\begin{center}
\includegraphics[scale=0.3]{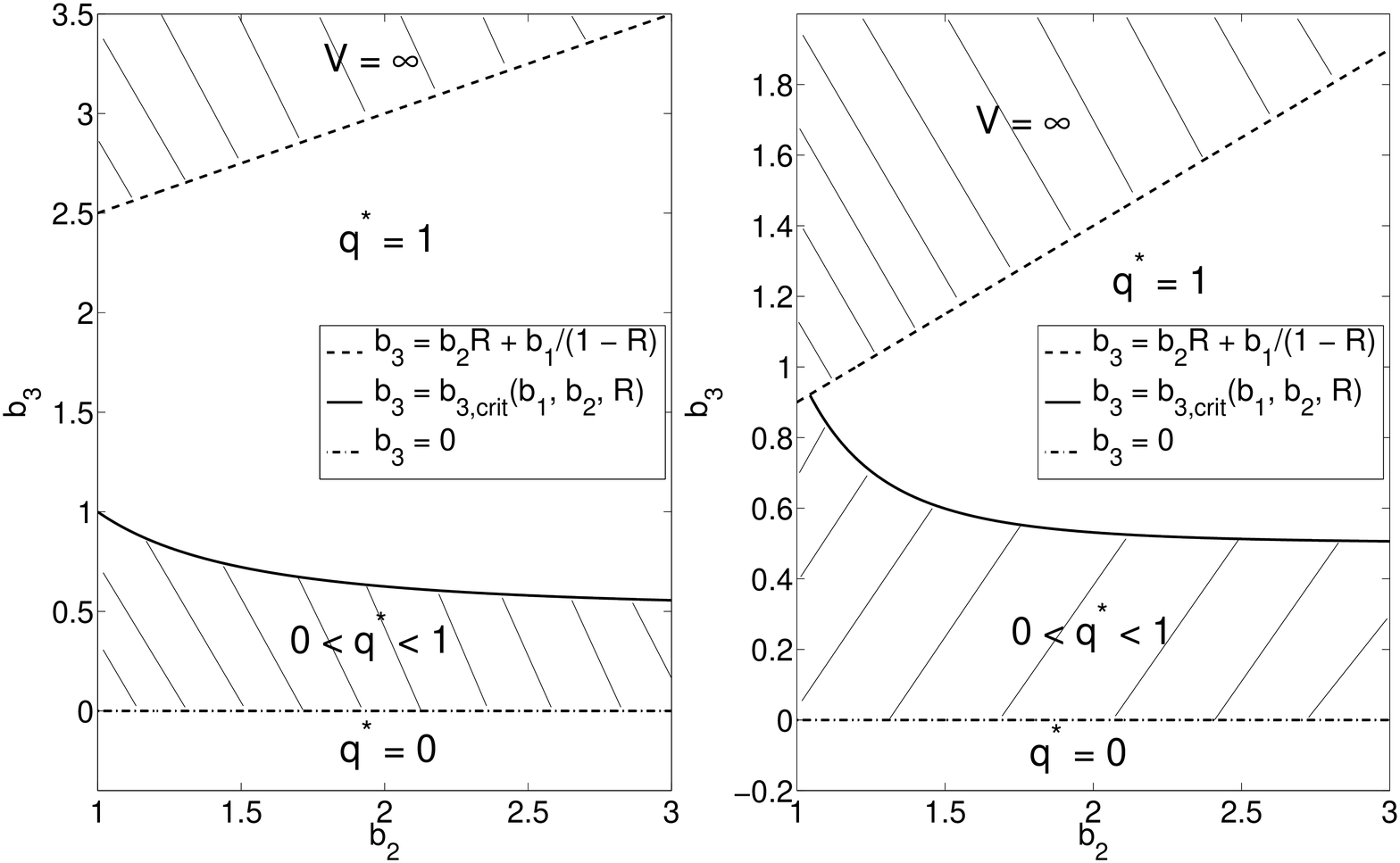}
\end{center}

\caption{Plot of the regions in $(b_2,b_3)$ space which correspond to different characteristics of solution, for fixed $b_1>0$ and $R<1$. In the left graph parameters are $b_1 = 1$ and $R = 0.5$. In the right graph parameters are $b_1 = 0.2$ and $R = 0.5$ so that $b_1<R(1-R)$ and the solid and dashed lines intersect at a value strictly above 1.}
\end{figure}

\begin{figure}
\begin{center}
\includegraphics[scale=0.3]{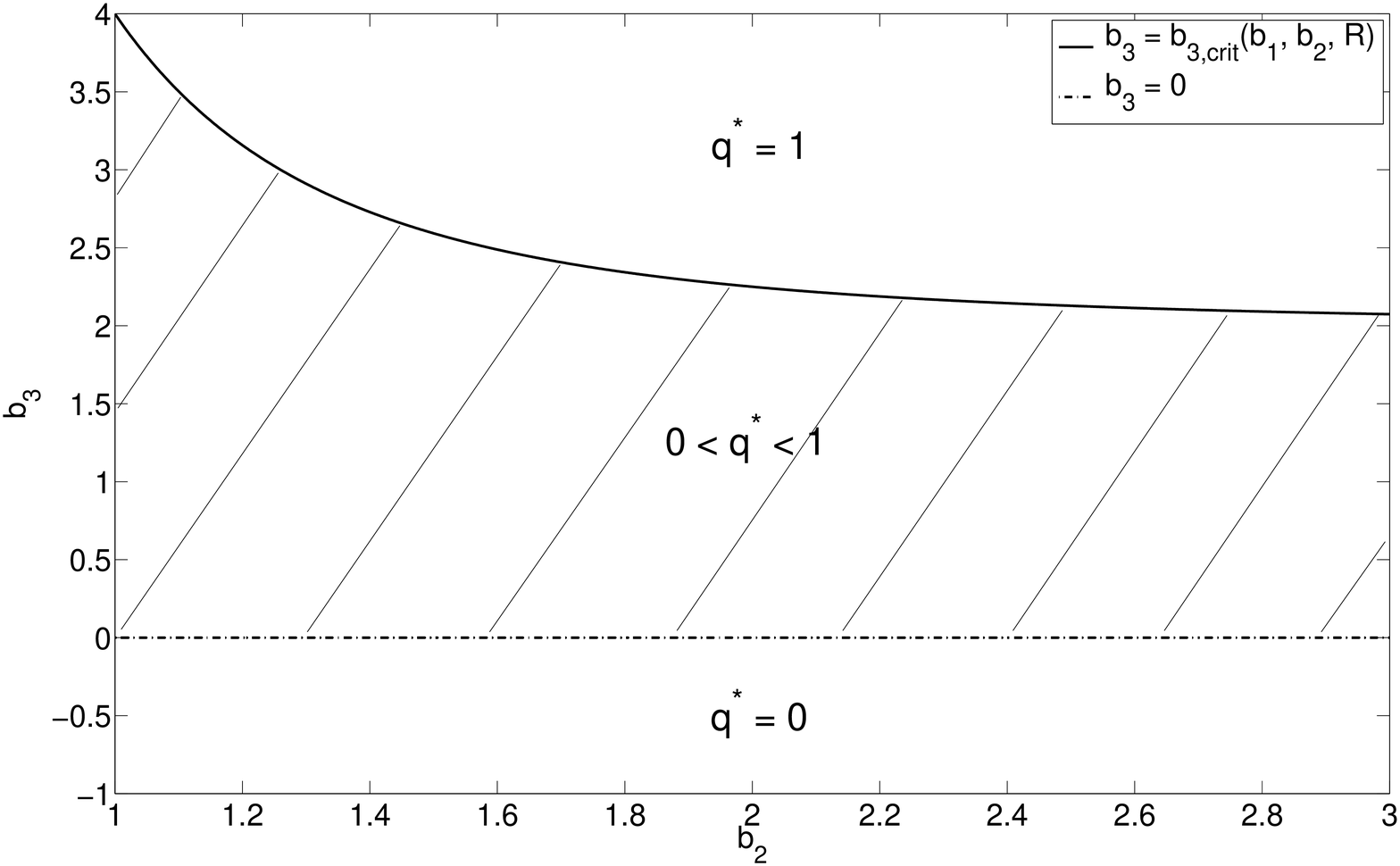}
\end{center}

\caption{Plot of the regions in $(b_2,b_3)$ space which correspond to different characteristics of solution, for fixed $b_1>0$ and $R>1$. Note that when $R>1$ there is no region where the value function is infinite (provided $b_1>0$). The line $b_{3,crit} = b_{3,crit}(b_1,b_2,R)$ which separates $q^*<1$ from $q^*=1$ has limiting value $2R$ at $b_2=1$ and $R$ at $b_2 = \infty$. The figure is drawn in the case $b_1=1$ and $R=2$.}
\end{figure}

The most interesting cases of Theorem~\ref{thm:4cases} are the middle two non-degenerate cases, and these two cases
we study in more detail in the next two theorems. Recall that we suppose we have constructed the solution
$n$ to the differential equation in (\ref{eq:n}).
Define $N(q) = n(q)^{-R} (1-q)^{R-1}$, and let $W$ be inverse to $N$. Let $h^* = N(q^*)$.

For a twice differentiable function $f$ define $\Psi_f(z)$ by
\begin{equation}
\label{eqn:Psidef}
\Psi_f(z) =  \frac{(1-R) f(z) - (1 + \frac{\eta \rho R}{\lambda}) zf'(z) - \frac{\eta \rho}{\lambda} z^2 f''(z)}{R(1-R) f(z) - 2Rzf'(z) - z^2 f''(z)}.
\end{equation}

\begin{thm}
\label{thm:maincase}
Suppose $R < 1$, and
suppose $0 < b_3 < b_{3,crit} (b_1, b_2, R)$,
so that $0 < q^* < 1$.

Then $N:[0,q^*]\mapsto[1,h^*]$ is increasing, and $W:[1,h^*]\mapsto[0,q^*]$ is well-defined and increasing.
Moreover $n(q^*)^{-R} =  h^{*} (1- q^*)^{1-R}$.

Let $z^{*}$ be given by
\begin{equation}
z^{*}=(1-q^*)^{-1}-1 = \frac{q^*}{1-q^*} \in
(0,\infty). \label{eq:886}
\end{equation}
and let $u^* = e^{z^*}$.
On $\left[1,h^{*}\right]$ let $h$ be the solution of
\begin{equation}
u^{*}-u=\int_{h}^{h^{*}}\frac{1}{(1-R)fW\left(f\right)}df.\label{eq:885}
\end{equation}
It follows that $h(-\infty) := \lim_{u \downarrow -\infty} h(u) = 1$.

Let $g$ be given by
\begin{equation}
g\left(z\right)=\begin{cases}
\begin{array}{l}
\left(\frac{b_1}{b_4 R}\right)^{-R} n(q^*)^{-R} \left(1+z\right)^{1-R}\\
\left(\frac{b_1}{b_4 R}\right)^{-R}h\left(\ln z\right)
\end{array} & \begin{array}{l}
\qquad z \in (z^{*}, \infty);\\
\qquad z \in [0, z^{*}].
\end{array}\end{cases}\label{eq:884}
\end{equation}
Then, the value function $V$ is given by
\begin{equation}
V\left(x,y,\theta,t\right)
 =  e^{-\beta t}\frac{x^{1-R}}{1-R}g\left(\frac{y\theta}{x}\right).
\hspace{10mm} x>0, y>0, \theta \geq 0
\label{eq:883}
\end{equation}
which by continuity extends to $x=0$ via
\[ V(0,y,\theta,t) = e^{-\beta t}\frac{y^{1-R} \theta^{1-R}}{1-R} \left(\frac{b_1}{b_4 R}\right)^{-R} n(q^*)^{-R} . \]

Let $(J,L) = (J_t, L_t)_{t \geq 0}$ be the unique pair such that
\begin{enumerate}
\item $J$ is positive,

\item $L$ is increasing, continuous, $L_0 = 0$, and $dL_t$ is
carried by the set $ \left\{ t:J_{t}= 0 \right\}$,

\item $J$ solves
\[ J_t = (z^*-z_0)^+ - \int_0^t \tilde{\Lambda}(J_s) ds
- \int_0^t \tilde{\Sigma}(J_s) dB^{1}_s
\int_0^t \tilde{\Gamma}(J_s) dB^{2}_s
+ L_t, \]
where
$\tilde{\Lambda}(j) =\Lambda(z^*-j)$, $\tilde{\Sigma}(j) =\Sigma(z^*-j)$ and
$\tilde{\Gamma}(j) = \Gamma(z^*-j)$, where in turn
\[ \Lambda(z) = z \left( g(z) - \frac{1}{1-R} z
g'(z)\right)^{-1/R} -\lambda(\lambda + \eta \rho) z \Psi_g(z) + \lambda^2 z \Psi_g(z)^2 + \zeta \eta z, \]
$\Gamma(z)=\eta z$ and $\Sigma(z) = - \sigma z \Psi_g(z)$.
\end{enumerate}
For such a pair $0 \leq J_t \leq z^*$.

Let $z_0= y_0 \theta_0 /x_0$. If $z_0 \leq z^*$ then set $\Theta^*_0 =
\theta_0$ and $X^*_0 = x_0$; else $z_0 > z^*$ and for the optimal strategy there is
a sale of
a positive
quantity $\theta_0 - \Theta_0$ of units at time $0$ such that
\[ \Theta^*_0 = \theta_0 \frac{z^*}{(1+z^*)}\frac{(1+z_0)}{z_0} \leq
\theta_0 \]
and $X^*_0 = x_0 + y_0(\theta_0 - \Theta_0)$.

Then, the optimal holdings $\Theta^*_t$
of the endowed asset, and the resulting wealth process are given by
\begin{eqnarray*}
\Theta^*_t  & = & \exp\left\{
-\frac{1}{z^*(1+z^*)} L_t \right\} ; \\
X^*_t & = & \frac{Y_t \Theta^*_t}{ (z^* - J_t)};
\end{eqnarray*}
and
the optimal consumption process
$C^*_t = C(X^*_t,Y_t,\Theta^*_t)$, the optimal portfolio process
$\Pi^*_t = \Pi(X^*_t,Y_t,\Theta^*_t)$,
and the certainty equivalent value $p$ are given in feedback form via
\begin{eqnarray}
C(x,y,\theta) & = & x \left[g\left( \frac{y\theta}{x} \right)
-\frac{1}{1-R} \frac{y\theta}{x}
g' \left( \frac{y\theta}{x} \right)\right]^{-\frac{1}{R}}  \label{eq:Cdef}  \\
\Pi(x,y,\theta) & = & \frac{\lambda}{\sigma} x \Psi_g \left( \frac{ y \theta }{x} \right)
\label{eq:Pidef}\\
p(x,y,\theta) & = & x \left[\frac{g\left(\frac{y\theta}{x}\right)}{g(0)}\right]^{\frac{1}{1-R}} - x.
\label{eq:CEdef}
\end{eqnarray}

Now suppose $R>1$ and $0 < b_3 < b_{3,crit} (b_1, b_2, R)$ so that $0 < q^* < 1$.
Let all quantities be defined as before. Then $N:[0,x^*]\mapsto[h^*, 1]$ is decreasing and
$W:(h^*, 1)\mapsto [0,q^*]$ is well defined and decreasing. On $(h^*, 1)$ $h$ is defined via
\[
u^{*}-u=\int_{h^*}^{h}\frac{1}{(R-1)fW\left(f\right)}df.
\]
The value function $V$, the optimal holdings $\Theta^*$, the optimal consumption process
$C^*$, the optimal portfolio process $\Pi^*$, the resulting wealth process $X^*$ and the
certainty equivalent value $p$ are the same as before.

\end{thm}

\begin{rem}
Given $n$ and the first crossing point $q^*$ the construction of $N$, $W$, $h$, $g$ and hence $V$ and $p$ is
immediate and straightforward.

Further, given realisations of the price processes $P$ and $Y$ (or equivalently paths of the Brownian
motions $B^1$ and $B^2$ (or $B^1$ and $B^{\perp}$) then $J$ and $L$ arise from a pathwise solution of a
Skorokhod problem~\cite[Lemma VI.2.1]{RY}. The optimal endowed asset holdings $\Theta^*$ and then also the optimal
cash wealth process $X^*$ are given explicitly in terms of the the solution of the Skorokhod problem;
the optimal consumption and investment are then given in feedback form as functions of these primary quantities.
\end{rem}

\begin{thm}
\label{thm:ndeg2}
Suppose $R<1$ and suppose $b_{3,crit}(b_1, b_2, R) \leq b_3 < \frac{b_1}{1-R} + b_2 R$.

Let $n$ solve {\rm (\ref{eq:n})} on $[0,1]$. Then for the given parameter combinations we have $q^*=1$.
Then $N$ is increasing and $W$ is well defined.
Define $\gamma:(1,\infty) \mapsto \mathbb R$ by
\begin{equation}
\label{eqn:gammadefR<1}
\gamma(v) = \frac{\ln v}{1-R} + \frac{R}{1-R} \ln n(1) - \frac{1}{1-R} \int_v^\infty \frac{(1-W(s))}{s W(s)} ds .
\end{equation}
Let $h$ be inverse to $\gamma$ and let $g(z) = (b_4 R/b_1)^R h (\ln z)$.

Then, the value function $V$ is given by
\begin{equation}
V(x,y,\theta,t) = e^{-\beta t}  \frac{x^{1 - R}}{1 - R}
g\left(\frac{y\theta}{x}\right), \hspace{10mm} x>0, y>0, \theta \geq 0
\label{eq:v}
\end{equation}
and $V(0,y,\theta,t) = e^{-\beta t}  \frac{y^{1-R} \theta^{1 - R}}{1 - R} n(1)^{-R}$.

Suppose $\theta_0>0$.
Let $K_0 = x_0/(y_0 \theta_0) \in [0,\infty)$. Let $(K,L) = (K_t, L_t)_{t \geq 0}$ be the unique pair such that
\begin{enumerate}
\item $K$ is positive,

\item $L$ is increasing, continuous, $L_0 = 0$, and $dL_t$ is
carried by the set $ \left\{ t: K_{t}= 0 \right\}$,

\item $K$ solves
\[ K_t = K_0 + \int_0^t \hat{\Lambda}(K_s) ds +
\int_0^t \hat{\Sigma}(K_s) dB^{1}_s +
\int_0^t \hat{\Gamma}(K_s) dB^{2}_s
+ L_t, \]
where
\[ \hat{\Lambda}(k) = (\eta - \zeta) \eta k + \lambda(\lambda - \eta \rho) k \Psi_g(1/k) - k \left[ g(1/k) - \frac{g'(1/k)}{k(1-R)} \right]^{-1/R}, \]
$\hat{\Sigma}(k) = \lambda k \Psi_g(1/k)$ and $\hat{\Gamma}(k) = - \eta k$.
\end{enumerate}

Then, the optimal holdings $\Theta^*_t$
of the endowed asset, and the optimal wealth process are given by
\begin{eqnarray}
\label{eq:optimalT}
\Theta^*_t  & = & \theta_0 \exp\left\{- L_t \right\}, \\
\label{eq:optimalX}
X^*_t & = &  Y_t \Theta^*_t K_t
\end{eqnarray}
The optimal consumption process
$C^*_t = C(X^*_t,Y_t,\Theta^*_t)$, the optimal portfolio process
$\Pi^*_t = \Pi(X^*_t,Y_t,\Theta^*_t)$,
and the certainty equivalent value $p=p(X^*_t,Y_t,\Theta^*_t)$ are given in feedback form by the expressions in
(\ref{eq:Cdef}), (\ref{eq:Pidef}) and (\ref{eq:CEdef}).

Now suppose $R>1$ and $b_3 \geq b_{3,crit}$. Then $N$ is decreasing. Define
\begin{equation}
\label{eqn:gammadefR>1}
\gamma(v) = -\frac{\ln v}{R-1} - \frac{R}{R-1} \ln n(1) - \frac{1}{R-1} \int_0^v \frac{(1-W(s))}{s W(s)} ds .
\end{equation}
Let $h$ be inverse to $\gamma$ and define $g$ and the value function as in the case $R<1$. Then the optimal holdings, optimal consumption, optimal investment in the financial asset, optimal wealth process and the certainty equivalent value of the holdings are all as given in the case $R<1$.

\end{thm}

\begin{rem}
Based on the results in Theorems~\ref{thm:maincase} and \ref{thm:ndeg2} we can show (Corollary~\ref{cor:pmonb}) that the certainty equivalent value of the holdings in the endowed asset are increasing in $b_3$ and decreasing in $b_1$ and $b_2$. Thus, for example, the certainty equivalent value is increasing in $\zeta$ the Sharpe ratio of the endowed asset, and decreasing in $\beta$ the discount factor.
\end{rem}

\begin{rem}
Here is one, perhaps surprising, consequence of Theorem~\ref{thm:ndeg2} which holds for $b_2>1$.
Define the stopping time $\tau = \inf \{t; X^*_t = 0\}$. Then, for the parameter combinations studied in Theorem~\ref{thm:ndeg2} and under optimal behaviour,
the investments in the risky asset are such that
$\Pi^*_{\tau} = \Pi_\tau(0, Y_\tau, \Theta^*_\tau) \neq 0$. This implies that even when liquid wealth is zero, it is not optimal to invest zero amount
in the financial asset. Adverse movements in the price of the financial asset have a negative impact on liquid wealth, and must be financed through
sales of the endowed asset. Conversely, beneficial movements in the price of the financial asset will generate positive liquid wealth for the agent. In particular,
$X^*=0$ is not an absorbing state.

In contrast, for $b_2=1$ and for the parameter combinations where Theorem~\ref{thm:ndeg2} applies, no sales of the endowed asset
take place until liquid wealth has been exhausted, but once liquid wealth is zero, there are no investments in the financial asset, cash wealth is maintained at zero,
and consumption is financed through sales of the endowed asset, continuously over time.
\end{rem}

\section{Proofs and verification arguments: the degenerate cases}

For $\sH=\sH(x,y,\theta) : [0,\infty) \times(0,\infty) \times [0,\infty) \mapsto \R$ with $\sH \in C^{2,2,1}$ and such that $\sH_x > 0$ define operators
\begin{eqnarray*}
\mathcal{L} \sH & = & \sup_{(c>0, \pi)}\left\{ \frac{c^{1-R}}{1-R} -c \sH_{x}+\alpha y\sH_{y}+\lambda\sigma\pi \sH_{x}+rx\sH_{x}+\frac{1}{2}\sigma^{2}\pi^{2}\sH_{xx}+\frac{1}{2}\eta^{2}y^{2}\sH_{yy}+\sigma\eta\rho y\pi \sH_{xy}\right\}\\
& = &  \sH_{x}^{1-1/R} \frac{R}{1-R} + rx \sH_x + \alpha y \sH_y + \frac{1}{2} \eta^2 y^2 \sH_{yy} - \frac{(\eta \rho y \sH_{xy} + \lambda \sH_x)^2}{2\sH_{xx}},\\
\mathcal{M} \sH  & = & \sH_{\theta}-y\sH_{x}.
\end{eqnarray*}
$\mathcal{L}\sH$ is defined on $(0, \infty)\times(0,\infty)\times[0,\infty)$.
However we can extend the domain of definition of $\sL$ to $[0, \infty)\times(0,\infty)\times[0,\infty)$ by extending the
definition of $\sH$ to the region $-y \theta<x \leq 0$ in such a way that the derivatives of $\sH$ are continuous
at $x=0$. $\mathcal{M}\sH$ is defined on $(0, \infty)\times(0,\infty)\times(0,\infty)$. Note that we will not need $\sM \sH$ at $\theta = 0$, but we can extend the domain of definition of
$\sM$ to $x=0$ using the same extension of $\sH$ to $x \leq 0$.

\subsection{The Verification Lemma in the case of a depreciating asset.}

Suppose $b_3 \leq 0$. Our goal is to show that the conclusions of
Theorem~\ref{thm:4cases}(1) hold.

From Proposition~\ref{prop:main} we know $q^* = 0$. Define the candidate value function via $G(x,y,\theta,t)
= e^{-\beta t} \sG(x,y,\theta)$ where
\begin{equation}
\sG(x,y,\theta) =  \left(\frac{b_1}{b_4 R}\right)^{-R} \frac{(x + y\theta)^{1-R}}{1-R}.
\label{eq:vdeg1}
\end{equation}
The candidate optimal strategy is to sell all units of the risky asset immediately.

Prior to the proof of the theorem, we need the following lemma.

\begin{lem}
\label{deg1:inequal}
Suppose $b_3 \leq 0$. Consider the candidate value function constructed in (\ref{eq:vdeg1}). Then $\mathcal{M} \sG = 0$, and $\mathcal{L} \sG- \beta \sG \leq 0$ with equality at
$\theta = 0$.
\end{lem}
\begin{proof}
Given the form of the candidate value function in (\ref{eq:vdeg1}), $\mathcal{M} \sG = 0$ follows immediately.
On the other hand, writing $z = y\theta/x$,
\[
\mathcal{L}\sG - \beta \sG = x^{1-R} R \left(\frac{b_1}{b_4 R}\right)^{1-R} (1+z)^{1-R}
\left[\frac{b_3}{b_1}\frac{z}{1+z} - \frac{R}{b_1} \left(\frac{z}{1+z}\right)^2 \right]
\leq 0,
\]
with equality at $z = 0$, which completes the proof.
\end{proof}

\begin{thm}
\label{thm:deg1}
Suppose $b_3 \leq 0$. Then the value function $V$ is given by
\begin{equation}
V(x,y,\theta,t) =  e^{-\beta t} \left(\frac{b_1}{b_4 R}\right)^{-R} \frac{(x + y\theta)^{1-R}}{1-R}= G(x,y,\theta,t).
\label{eq:vsimple}
\end{equation}
The optimal holdings $\Theta^*_t$ of the endowed asst, the optimal consumption process $C^*_t$ and the resulting wealth process are given by
\[ 
(\bigtriangleup \Theta^{*})_{t = 0} = -\theta_0,
\hspace{6mm}
C^{*}_t = \frac{b_1}{b_4 R} X^*_t,
\hspace{6mm}
\Pi^*_t = \frac{\lambda}{\sigma R} X^*_t,
\] 
\begin{equation}
X^*_t = (x_0 + y_0 \theta_0) \exp \left\{\left(\frac{\lambda^2}{R} + r - \frac{b_1}{b_4 R} - \frac{\lambda^2}{2R^2}\right)t + \frac{\lambda}{R}B^1_t\right\}.
\label{eq:optmal33}
\end{equation}
The certainty equivalence price is given by $p(x,y,\theta,t) = y\theta$.
\end{thm}

\begin{proof}
We prove the result at $t=0$, ie we show that $\sV(x,y,\theta)=V(x,y,\theta,0) = G(x,y,\theta,0) = \sG(x,y,\theta)$; the case of general $t$ follows from the time-homogeneity of the problem.

Note that under proposed strategies in (\ref{eq:optmal33}), the optimal strategy is to sell the entire endowed asset holding immediately, which gives
$X^*_{0+} = x_0 + y_0\theta_0$ and to finance investment and consumption from liquid wealth thereafter.
Hence, the wealth process $(X^*_t)_{t \geq 0}$ evolves as $dX^*_t = (\lambda \sigma \Pi^*_t + rX^*_t - C^*_t)dt + \sigma \Pi^*_t dB^1_t$. This gives
$X^*_t =(x_0 + y \theta_0) \exp \left\{\left(\frac{\lambda^2}{R} + r - \frac{b_1}{b_4 R} - \frac{\lambda^2}{2R^2}\right)t + \frac{\lambda}{R}B^1_t\right\}$.

The value function under strategy proposed in (\ref{eq:optmal33}) is
\begin{eqnarray*}
\lefteqn{\mathbb{E} \left[\int_0^\infty e^{-\beta t} \frac{{C^{*}_t}^{1-R}}{1-R} dt \right]} \\
 & = &
\left(\frac{b_1}{b_4 R}\right)^{1-R} \frac{(x_0 + y_0 \theta_0)^{1-R}}{1-R} \int_0^\infty  \exp\left\{ \left[ \frac{\lambda^2(1-R)}{2R} + r(1-R) - \frac{b_1(1-R)}{b_4 R} - \beta \right] t \right\}dt \\
& = &
\left(\frac{b_1}{b_4 R}\right)^{1-R} \frac{(x_0 + y_0 \theta_0)^{1-R}}{1-R} \int_0^\infty  \exp\left\{ \left( -\frac{b_1}{b_4 R} \right) t \right\}dt \\
& = & \left(\frac{b_1}{b_4R}\right)^{-R} \frac{(x_0 + y_0\theta_0)^{1 - R}}{1 - R} = \sG(x_0+y_0\theta_0,y_0,0).
\end{eqnarray*}
Hence $V(x,y,\theta,0) \geq G(x,y,\theta,0)$.

Now, consider general admissible strategies. Suppose first that $R<1$. Define the process
$M = (M_t)_{t \geq 0}$ by
\[ 
M_t = \int_0 ^t e^{-\beta s} \frac{C_s ^{1-R}}{1-R}ds + e^{-\beta t}\sG\left(X_t, Y_t, \Theta_t\right).
\] 
Applying the generalised It\^{o}'s formula~\cite[Section 4.7]{gito} to $M_t$ leads to
\begin{eqnarray}
\nonumber
M_t - M_0 & = & \int_0^t  e^{-\beta s} \Big[ \frac{C_s^{1-R}}{1-R} -C_s \sG_{x}+\alpha Y_s \sG_{y} +\lambda\sigma\Pi_s \sG_{x}+rX_s \sG_{x} \\
\nonumber
& & \hspace{10mm} +\frac{1}{2}\sigma^{2}\Pi_s^{2} \sG_{xx}+\frac{1}{2}\eta^{2}Y_s^{2}\sG_{yy}+\sigma\eta\rho Y_s\Pi_s \sG_{xy} - \beta \sG\Big]ds \\
\nonumber
& & + \int_0^t e^{-\beta s} (\sG_\theta - Y_s \sG_x) d\Theta_s \\
\label{eq:Mexpansion}
& & + \sum_{\substack{0<s \leq t}} e^{-\beta s}
\left[\sG(X_s, Y_s, \Theta_s) - \sG(X_{s-}, Y_{s-}, \Theta_{s-}) - \sG_x (\bigtriangleup X)_s
- \sG_\theta (\bigtriangleup \Theta)_s \right] \\
\nonumber
& & +  \int_0^t e^{-\beta s} \sigma \Pi_s \sG_x dB^1_s \\
\nonumber
& & + \int_0^t e^{-\beta s} \eta Y_s \sG_y dB^2_s \\
\nonumber
& = & N_t^1 + N_t^2 + N_t^3 + N_t^4 + N_t^5.
\end{eqnarray}
Note that for a general admissible strategy $\Theta$ and $X$ do not need to be continuous, so
that here the arguments of $\sG_{\cdot}$ are
$(X_{s-},Y_{s},\Theta_{s-})$.

Lemma~\ref{deg1:inequal} implies that $\mathcal{L}\sG - \beta \sG\leq 0$ and $\mathcal{M} \sG = 0$, which leads to
$N_t^1 \leq 0$ and $N_t^2 = 0$. Using the fact that $(\Delta X)_s = -Y_s (\Delta \Theta)_s$ and
writing $\theta = \Theta_{s-}$, $x = X_{s-}$, $\chi = - (\Delta
\Theta)_s$ each non-zero jump in $N^3$ is of the form
\[
(\Delta N^3)_s = e^{-\beta s} \left\{ \sG(x+y \chi,y,\theta - \chi) - \sG(x,y,\theta) +
\chi \left[ \sG_\theta (x,y,\theta) - y \sG_x (x,y,\theta) \right] \right\}.
\]
Given the form of the candidate value function in (\ref{eq:vdeg1}), it is easy to see that
 $\sG(x+y \phi,y,\theta - \phi)$ is constant in $\phi$, whence
 $y \sG_x = \sG_{\theta}$ and $(\Delta N^3) = 0$.
Then, since $R<1$, we have $0\leq M_t \leq M_0 + N_t ^4 + N_t ^ 5$, and
$(N^4 + N^5)_{t \geq 0}$, as the sum of two local martingales, is a local martingale and is bounded from below and hence a supermartingale. By taking expectations we find $\mathbb{E}(M_t) \leq M_0 = G(x_0, y_0, \theta_0, 0)$, which gives
\[
G(x_0, y_0, \theta_0, 0) \geq \mathbb{E} \int_0^t e^{-\beta s} \frac{C_s ^{1-R}}{1-R}ds +
\mathbb{E} \left[ e^{-\beta s} \sG(X_t, Y_t, \Theta_t)\right] \geq \mathbb{E} \int_0^t e^{-\beta s} \frac{C_s ^{1-R}}{1-R}ds,
\]
where the last inequality follows since $\sG(X_t, Y_t, \Theta_t) \geq 0$ for $R\in(0,1)$. Letting $t \to \infty$, we have
\[
G(x_0, y_0, \theta_0, 0) \geq \mathbb{E}\left(\int_0^\infty e^{-\beta t} \frac{C_t^{1-R}}{1-R}dt\right),
\]
and taking a supremum over admissible strategies leads to $G(x_0, y_0, \theta_0, 0) \geq V(x_0, y_0, \theta_0, 0)$.

The case $R>1$ will be considered in Appendix~\ref{app:R>1}.

\end{proof}

\subsection{Proof of the ill-posed case of Theorem~\ref{thm:4cases}.}
In the case $R<1$ and $b_3 \geq \frac{b_1}{1-R} + b_2 R$
it is sufficient to give an example of an admissible strategy for which the expected utility of
consumption is infinite.

The condition $b_3 \geq \frac{b_1}{1-R} + b_2 R$ can be rewritten in terms of the original parameters as
\begin{equation}
\zeta \eta (1-R) - \frac{1}{2} \eta^2 R(1-R) + r(1 - R) - \beta \geq 0.
\label{eq:pa}
\end{equation}
Consider the following consumption, investment and sale strategies
\begin{equation}
\tilde{C}_t = \lambda \eta \tilde{X}_t,
\hspace{6mm}
\tilde{\Pi}_t = \frac{\eta}{\sigma} \tilde{X}_t,
\hspace{6mm}
\frac{d\tilde{\Theta}_t}{\tilde{\Theta}_t} = - \zeta \eta \frac{\tilde{X}_t}{Y_t \tilde{\Theta}_t} dt.
\label{eq:illstra}
\end{equation}

Note first that $\tilde{\Theta}$ is a non-increasing process. The corresponding wealth process $\tilde{X}$ satisfies $d\tilde{X}_t = (\zeta \eta + r)\tilde{X}_t dt + \eta \tilde{X}_t dB_t ^1$,
which gives
\[
\tilde{X}_t = x_0 \exp \left\{\left(\zeta \eta + r - \frac{1}{2}\eta^2 \right)t + \eta B^1_t\right\}.
\]
Hence $\tilde{X}_t \geq 0$ and the strategies defined in (\ref{eq:illstra}) are admissible.

The expected utility from consumption $\tilde{G}$ corresponding to the consumption, sale and investment strategies $(\tilde{C}, \tilde{\Theta}, \tilde{\Pi})$ is given by
\begin{eqnarray*}
\tilde{G} & = &\mathbb{E}^{(\tilde{c}, \tilde{\theta}, \tilde{\pi})}
\left[ \int_0^ \infty e^{-\beta t}  \frac{{\tilde{C}_t}^{1-R}}{1-R} dt \right]  \\
& = & \frac{\left(\lambda \eta x_0 \right)^{1-R}}{1-R} \int_0^\infty \exp \left\{\left[
\zeta \eta (1-R) - \frac{1}{2} \eta^2 R(1-R) + r(1 - R) - \beta
\right]t\right\} dt \\
& = & \infty,
\end{eqnarray*}
where the last equality follows from (\ref{eq:pa}).

\section{Some preliminaries on $n$ and $q^*$.}

Recall the definitions of $\varphi$ and $\upsilon$:
\begin{eqnarray*}
\upsilon(q,n) & = & \varphi(q,n) - \sgn(1-R) \sqrt{\varphi(q,n)^2 + E(q)^2}, \\
\varphi(q,n) & = & b_1 (n-1) + (1-R)(b_3 - 2R)q + (2 - b_2) R(1-R);
\end{eqnarray*}
where $E(q)^2 = 4 R^2(1-R)^2 (b_2-1)(1-q)^2 \geq 0$.

Results not proved in the main body of this section are proved in Appendix~\ref{app:n}.

\begin{lem} For $q \in [0,1]$, $(1-R)q + R > 0$ and
\begin{eqnarray*}
\varphi(q,m(q)) & = & R(1-R) \left\{ (1-q)^2 - (b_2 - 1) \right\} \\
\varphi(q,\ell(q)) & = & R(1-q) \left\{ (1-R)q + R  - \frac{(b_2 - 1)R^2}{(1-R)q + R} \right\} \\
\upsilon(q,m(q)) & = & - 2R(1-R)(b_2 - 1) \\
\upsilon(q,\ell(q)) & = & \frac{-2 R^2 (1-R) (1-q) (b_2-1)}{(1-R)q + R}
\end{eqnarray*}
Suppose $b_2>1$. Then, for $R<1$, $\upsilon(q,n)<0$ and for fixed $q$, $\upsilon(q,n)$ is an increasing, concave function of $n$.
For $R>1$, $\upsilon(q,n)>0$ and for fixed $q$, $\upsilon(q,n)$ is an increasing, convex function of $n$.

For $b_2=1$, $\upsilon(q,n)=0$.
\label{lem:varphi}
\end{lem}

Recall that $n$ solves (\ref{eq:n}). This expression can be written in several ways.

\begin{lem}
\label{lem:n'}
$n'(q) = F(q,n(q))$ where
\begin{eqnarray}
\nonumber\lefteqn{F(q,n)} \\
 & = & n \left( \frac{1-R}{R(1-q)} - \frac{(1-R)^2}{b_1 R}\frac{q}{\ell(q) - n} +
\frac{(1 - R)q}{2b_1 R (1-q) [(1-R)q + R]} \frac{\upsilon(q, n)}{\ell(q) - n} \right) \label{eq:nA}  \\
       & = &  \frac{(1-R)n}{R(1-q)} - \frac{ 2(1-R)^2 q n/R }{2(1-R)(1-q)[(1-R)q + R] -
                     \varphi(q,n) - sgn(1-R)\sqrt{\varphi(q,n)^2 + E(q)^2}}  \label{eq:nB}  \\
       & = &  - \frac{(1-R) n \left[ 2 b_1 [ (1-R)q+ R ](n-m(q)) - q (\upsilon(q,n) - \upsilon(q,m(q))) \right]}{
                  2R(1-q) ( (1-R)q+ R ) \{ S(q) - b_1 (n - m(q)) \} }  \label{eq:nC}
\end{eqnarray}
where $S(q) = b_1 (\ell(q) - m(q)) = (1-R)q(1-q) + (b_2-1) R (1-R) q /( (1-R)q+ R )$.
\end{lem}

\begin{lem}
\label{lem:n}
\begin{enumerate}
\item For $R \in (0,1)$, $n'(0)$ is the smaller root of $\Phi(\chi) = 0$ , where
\[
\Phi(\chi) = b_1 R \chi^2 + R(1-R)(b_3 - b_2 - \frac{b_1}{R}) \chi - b_3 (1-R)^2,
\]
For $R \in (1, \infty)$, $n'(0)$ is the larger root.

\item For $q \in (0,q_n)$, $(1-R)n(q) \leq (1-R)\ell(q)$.

\item For $q \in (0, q_n)$, $n'(q) < 0$ if and only if $n(q) > m(q)$,
$n'(q)>0$ if and only if $n(q)<m(q)$, and $n'(q)=0$ if and only if $n(q)=m(q)$.

\item If $R>1$ then either $q_n = q^*=1$ or $q_n > q^*$. If $R<1$ and $\ell(1) \geq 0$ then $q_n = q_{\ell} = 1$.
If $R <1$ and $\ell(1) < 0$ then $q_n = q_{\ell} < q^*$.

\item If $0 \leq q^* < 1$ then $q^* > \frac{b_3}{2R}$ and $(1-R)m$ is
increasing on $(q^*,1)$.

\end{enumerate}
\end{lem}

\begin{lem}
\label{lem:monbi}
$(1-R) n$ is monotone increasing in $b_1$, monotone increasing in $b_2$ and monotone decreasing in $b_3$ for $q \in [0,q^*]$. Moreover $q^*$ is decreasing as a function of $b_1$, decreasing as a function of $b_2$ and
increasing as a function of $b_3$. Finally $(1-R)n(q^*)^{-R}$ is decreasing in $b_1$, decreasing in $b_2$ and increasing in $b_3$.
\end{lem}

We have not yet made the connection between $n$ and the value function, but accepting for the moment the relationships in (\ref{eq:884}), (\ref{eq:883}) and (\ref{eq:v}) this lemma has the following important corollary.

\begin{cor}
\label{cor:pmonb}
The certainty equivalent value $p=p(x,y,\theta)$ is decreasing in $b_1$ and $b_2$ and increasing in $b_3$.
\end{cor}

\begin{lem}
\label{lem:b3critbound}
\begin{enumerate}
\item Suppose $b_2 = 1$.  Then $b_{3,crit}(b_1, 1, R) = \bar{b}_3$ where recall $\bar{b}_3 = \min \{ 2R, R + \frac{b_1}{1-R} \}$ if $R<1$ and $\bar{b}_3 = 2R$ if $R>1$.
\item Let $b_2 \to \infty$. Then $\lim_{b_2 \to \infty} b_{3,crit} (b_1, b_2, R) = R$.
\end{enumerate}
Hence $R < b_{3,crit} \leq \bar{b}_3$.
\end{lem}

\begin{proof}[Proof of Proposition \ref{prop:main}.]
Suppose $R<1$. (The proof for $R>1$ is similar, except that the final paragraph is not necessary.)

Recall the definition of $\Phi$ in Lemma~\ref{lem:n}.
Note that $\Phi(m'(0)) = (1-R)^2Rb_2b_3$. Then, if
$b_3<0$ we have $n'(0)<m'(0)$ and $q^*=0$.
If $b_3 = 0$ then $n'(0)=m'(0)$ and
more care is needed.
Since $b_3 \leq 0$, $m$ is increasing. Suppose
$n\left(\hat{x}\right)>m\left(\hat{x}\right)$ for some $\hat{x}$
in $\left[0,1\right].$ Let $\underline{x}=\sup\left\{
x<\hat{x}:n\left(x\right)=m\left(x\right)\right\} $. Then on
$\left(\underline{x},\hat{x}\right)$ we have
$n'\left(x\right)<0<m'\left(x\right)$ and
$m\left(\hat{x}\right)-n\left(\hat{x}\right)
=m\left(\underline{x}\right)-n\left(\underline{x}\right)
+\int_{\underline{x}}^{\hat{x}}
[m'\left(y\right)-n'\left(y\right)]dy>0$,
a contradiction. Hence $q^*=0$.

Now suppose $b_3>0$. Then $n'(0)> m'(0)$ and at least initially $n>m$, and $q^*>0$.
By Lemma~\ref{lem:monbi}, there exists a critical value of $b_3$, namely $b_{3,crit}$, such that $q^*(b_3)<1$ if $b_3 < b_{3,crit}$, and $q^*=1$ for $b_3 \geq b_{3,crit}$.

Note that if $b_3 \geq \frac{b_1}{1-R} + b_2 R$ then $\ell(1) \leq 0$. The results in (4) now follow from Lemma~\ref{lem:n}(4). Conversely, if $b_3 \geq 2R$ then $m$ is decreasing and $q^* \geq q_m$.

\end{proof}

Recall the definition $N(q) = n(q)^{-R}(1-q)^{R-1}$ and that $W$ is inverse to $N$. We have $h^* = N(q^*)$. Define $w(s)=(1-R)sW(s)$.
We close this section with some useful results about $W$ and $w$.

\begin{prop}
\label{prop:NWw}
(1) For $R<1$, $N$ is increasing on $[0,q^*]$. Moreover $W$ is increasing
and such that $0 < W(v) < q^*$ on $(1,h^*)$.
For $R>1$, $N$ is decreasing on $[0,q^*]$.
and $W$ is decreasing such that $0 < W(v) < q^*$ on $(h^*, 1)$.

(2) Let $w(s) = (1-R)sW(s)$.
Then $w$ satisfies $w(h^*) =(1-R)h^*W(h^*) = (1-R) h^* q^*$ and
solves
\begin{eqnarray}
 (w(s)w'(s))^2 + \left\{b_1 \left[s - \frac{1}{1-R}w(s)\right]^{\frac{R-1}{R}} - [b_1 + b_2R(1-R)]s + (b_3 + 2R - 2)w(s)\right\} w(s)w'(s) && \nonumber \\
  + [(2R-1) (b_3 - 1) + R^2(1 - b_2)]w(s)^2 + [(1 - 2R) b_1 + b_2R(1-R) - R(1-R)b_3]sw(s) && \nonumber \\
   + b_1R(1-R)h^2 + b_1[(2R-1) w(s) - R(1-R)s]\left[s - \frac{1}{1-R}w(s)\right]^{\frac{R-1}{R}} & = & 0 . \nonumber
\end{eqnarray}

(3) For $R<1$ and $1 < s < h^*$, and for $R>1$ and $h^* < s < 1$ we have $ 1-R <
w'(s)
< 1 - R w(s)/((1-R)s)$ with $w'(h^*)
= 1 - R w(h^*)/((1-R)h^*)$.

\end{prop}

\section{The verification lemma in the first non-degenerate case with finite critical exercise ratio}
Suppose $0 < b_3 < b_{3,crit} (b_1, b_2, R)$,
so that $0 < q^* < 1$. Suppose we have constructed $n$ and $N$ on $[0,q^*]$ and $W$ and $w$ on $[1,h^*]$. Set
$z^* = q^*/(1-q^*)$ and $u^* = e^{z^*}$.
Define $h$ via $h(u^*)=h^*$ and for $-\infty < u < u^*$  set $\frac{dh}{du} = w(h)$. Then also $\frac{d^2 h}{du^2} = w'(h) w(h)$. Then $h $ solves (\ref{eq:885}).
Define $g$ via (\ref{eq:884}).

\begin{lem}
\label{lem:smoothfit}
Let $z^*$ and $g$ be as given in Equations~\ref{eq:886} and
\ref{eq:884}
of
Theorem~\ref{thm:maincase}.
Then,
$g\left(z\right)$, $g'\left(z\right)$, $g''\left(z\right)$
are continuous at $z=z^{*}$.
\end{lem}
\begin{proof}
These results follow from the definition of $g$. For the smooth fit of first and second order, it is convenient to
note that $dh/du = w(h)$ and then for $z<z^*$,
\begin{eqnarray}
\nonumber
zg'(z) & = & \left(\frac{b_1}{b_4 R}\right)^{-R} h'(\ln z) =
\left(\frac{b_1}{b_4 R}\right)^{-R} w(h) \\
z^2 g''(z) = \left(\frac{b_1}{b_4 R}\right)^{-R} (h'' - h') =
\left(\frac{b_1}{b_4 R}\right)^{-R} (w'(h) - 1)w(h)  \label{eqn:2OSF}
\end{eqnarray}
The results now follow from the expressions for $w(h^*)$ and $w'(h^*)$ given in Proposition~\ref{prop:NWw}.


\end{proof}

\begin{prop}
\label{prop:gprop}
Suppose g solves (\ref{eq:884}). Then for $R<1$, $g$ is an increasing concave function such that $g(0) = \left(\frac{b_1}{b_4R}\right)^{-R}$.
Otherwise, for $R>1$, $g$ is a decreasing convex function such that $g(z) \geq 0$. Further, for
$z \leq z^*$ we have $0 \geq R g'(z)^2 \geq (1-R)g(z) g''(z)$ with equality at $z = z^*$.
\end{prop}

\begin{proof}
Consider first $R<1$. Since the statements are immediate in the region $z \geq z^*$, and
since there is second order smooth fit at $z^*$ the result will
follow if $h(-\infty)=1$, $h$ is increasing and, using the second part of (\ref{eqn:2OSF}),
$w(h) w'(h) - w(h) \leq 0$. The last two properties follow from
Proposition~\ref{prop:NWw} since $w(h)
\geq 0$ and $w'(h) < 1$.

To see that $h(-\infty)=1$ note that $w(h)=(1-R)hW(h)$ is bounded away from zero for $h$ bounded away from $1$. Then since $W'(1)$ is non-zero finite we conclude that $h(-\infty)=1$.

For $R>1$, and $z \geq z^*$, the statement holds immediately. For $z \leq z^*$,
 Proposition~\ref{prop:NWw} implies that
$h$ is decreasing and $w(h) \leq 0$, $w'(h) > 1$. Together with (\ref{eqn:2OSF}),
we have $g$ is a decreasing convex function and $g(z) \geq 0$ given that $h \in [0,1]$.
\end{proof}

Define the candidate value function at $t=0$ by
\begin{equation}
\sG(x,y,\theta) =  \frac{x^{1-R}}{1-R} g\left(\frac{y\theta}{x}\right).
\label{eq:vf}
\end{equation}

\begin{lem}
\label{lem:concave}
Fix $y$. Then $\sG=\sG(x,\theta)$ is concave in $x$ and
$\theta$. In particular, if $\psi(\chi) = \sG(x - \chi y \phi, y ,
\theta
+ \chi \phi)$, then $\psi$ is concave in $\chi$.
\end{lem}

\begin{proof}
In order to show the concavity of the candidate value function
it is sufficient
to show that the Hessian matrix given by
\[
H_{\sG}=\left(\begin{array}{cc}
\sG_{xx} & \sG_{x\theta}\\
\sG_{x\theta} & \sG_{\theta\theta}
\end{array}\right).
\]
has a positive determinant, and that one of the diagonal entries is
non-positive.

Direct computation leads to
\begin{eqnarray*}
\sG_{xx}\left(x,y,\theta\right) & = &
x^{-R-1}\left[-Rg\left(z\right)
+\frac{2R}{1-R}zg'\left(z\right)+\frac{1}{1-R}z^{2}g''
\left(z\right)\right],\\
\sG_{x\theta}\left(x,y,\theta\right) & = &
-x^{-R-1}\frac{y}{1-R}\left[Rg'\left(z\right)
+zg''\left(z\right)\right],\\
\sG_{\theta\theta}\left(x,y,\theta,t\right) & = &
x^{-R-1}\frac{y^{2}}{1-R}g''\left(z\right),
\end{eqnarray*}
and the determinant of Hessian matrix is
\begin{equation}
\label{eqn:Hess}
\sG_{xx}\sG_{\theta\theta}-\left(\sG_{x\theta}\right)^{2}
=  -x^{-2R}\theta^{-2}\frac{R}{1-R}
\left[g\left(z\right)z^{2}g''\left(z\right)
+\frac{R}{1-R}\left(zg'\left(z\right)\right)^{2}\right].
\end{equation}
If $z \geq z^*$ then the expression on the right-hand-side of
(\ref{eqn:Hess}) is zero by Proposition~\ref{prop:gprop}.
For $z \leq z^*$, Proposition~\ref{prop:NWw} yields
\begin{equation}
\label{eq:Hessian}
(1-R)g\left(z\right)z^{2}g''\left(z\right)+R
\left(zg'\left(z\right)\right)^{2}
=(1-R)  h\left[w\left(h\right)w'\left(h\right)-w\left(h\right)\right]
+R w\left(h\right)^{2}
\leq 0
\end{equation}
with equality at $h=h^{*}$ by the smooth fit condition.
Further, since $g$ is concave we have that $\sG_{\theta \theta} \leq
0$.

In order to show the concavity of $\psi$
in $\chi$, it is equivalent to examine the sign of
$\frac{d^2 \psi}
{d \chi^{2}}$. But
\[ \frac{d^2 \psi}
{d \chi^{2}}=\phi^{2}
\left[y^{2}\sG_{xx}+\sG_{\theta\theta}-2y\sG_{x\theta}\right] = \phi^2
(y,1)\det(H_\sG) (y,1)^T \leq 0\].

\end{proof}

\begin{lem}
\label{lem:operator}
Consider the candidate function constructed in (\ref{eq:vf}). Then

(a) For $0 < x \leq y \theta/z^*$, $\sM \sG=0$ and $\sL \sG  - \beta \sG \leq 0$.

(b) For $x \geq y \theta/ z^*$, $\sL \sG - \beta \sG =0$ and $\sM \sG \geq 0$.

Note that the derivatives of $\sG$ are well-defined and continuous at $x=0$, so that the results in (a) hold at $x=0$.
\end{lem}

\begin{proof}
(a) For $z \geq z^*$, $\sM \sG=0$ is immediate from the definition of $\sG$.
For
$\sL \sG - \beta \sG$ we have
that $\sG(x,y,\theta) = \left(\frac{b_1}{b_4 R}\right)^{-R}
n(q^*)^{-R}  \frac{x^{1-R}}{1-R} \left(1+z\right)^{1-R}$ and then
\[
\mathcal{L}\sG - \beta \sG 
=
 \frac{x^{1-R}}{1-R} R \left(\frac{b_1}{b_4 R}\right)^{1-R} n(q^*)^{-R} (1+z)^{1-R}
\left[m(q^*) - m\left(\frac{z}{1+z}\right)\right],
\]
where we use the fact that $n(q^*) = m(q^*)$. The required inequality
follows from Lemma~\ref{lem:n} and the fact that $m$ is increasing on $(q^*,1)$.

(b) In order to prove $\sL \sG - \beta \sG=0$ we calculate
\begin{eqnarray*}
\sL \sG - \beta \sG & = &  \frac{x^{1-R}}{1-R}
\left\{ \rule{0mm}{8mm}
(z^2g''(z))^2  \right. \\
&& \hspace{2mm}
+ \left[b_4R\left[g - \frac{1}{1-R} zg'(z)\right]^{\frac{R-1}{R}} - [b_1 + b_2R(1-R)]g
+ (b_3 + 2R)zg'(z)\right] z^2 g''(z)  \\
&& \hspace{2mm}
+ b_4 R [2Rzg'(z) - R(1-R)g] \left[g - \frac{1}{1-R} zg'(z)\right]^{\frac{R-1}{R}}
- [2Rb_1 + R(1-R)b_3]gzg'(z)\\
& & \hspace{2mm} + \left[2Rb_3 - R^2b_2 +R^2\right](zg'(z))^2 + b_1 R(1-R)g^2 \Bigg\} \\
& = &  \frac{x^{1-R}}{1-R} \left(\frac{b_1}{b_4 R}\right)^{-2R}
\left\{
(w(h)w'(h))^2 + \left(
b_1 \left[h - \frac{1}{1-R}w(h)\right]^{\frac{R-1}{R}} - \left[b_1 + b_2R(1-R)\right]h \right.\right. \\
& & \hspace{2mm}
+ \left(b_3 + 2R - 2\right)w(h)\Bigg) w(h)w'(h)  + \left[(2R-1) (b_3 - 1) + R^2(1 - b_2)\right]w(h)^2\\
& & \hspace{2mm}
 + \left[(1 - 2R) b_1 + b_2R(1-R) - R(1-R)b_3\right]hw(h) + b_1R(1-R)h^2 \\
& & \hspace{2mm}
+ b_1\left[(2R-1) w(h) - R(1-R)h\right] \left[h - \frac{1}{1-R}w(h)\right]^{\frac{R-1}{R}} \Bigg\} \\
& = & 0,
\end{eqnarray*}
where the last equality follows from Proposition~\ref{prop:NWw}(2). Note that $\sG$ and $\sL \sG$ are well-defined and continuous at $\theta = 0$.

Now consider $\sM \sG$. We have
\[
\mathcal{M} \sG
=  x^{-R}y\left[\frac{(1+z)}{1-R}g'
\left(z\right) - g\left(z\right)\right].
\]
Hence it is sufficient to show that
$\psi(z) \geq 0$ on $(0,z^*]$ where
\[
\psi\left(z\right)=\frac{1+z}{1-R}-\frac{g\left(z\right)}{g'
\left(z\right)},
\]
It follows from value matching and smooth fit
that $\psi(z^*)=0$ and hence it
it is sufficient to show that $\psi$ is decreasing.
But
\[ 
\psi^{'}\left(z\right)
=  \frac{R}{1-R}
+\frac{g\left(z\right)g^{''}\left(z\right)}{g^{'}\left(z\right)^{2}}
=
\frac{R}{1-R}
+\frac{h\left[w\left(h\right)w^{'}\left(h\right)-w\left(h\right)\right]}
 {w\left(h\right)^{2}}
\leq  0
\] 
where the last inequality follows from the final part of Proposition~\ref{prop:NWw}.

\end{proof}

\begin{prop}
\label{prop:Z}
Let $X^*$, $\Theta^*$, $C^*$ and $\Pi^*$ be as defined in Theorem~\ref{thm:maincase}.
Then they correspond to an admissible wealth process. Moreover $Z_t^* = Y_t
\Theta^*/X^*_t$ satisfies $0 \leq Z^*_t \leq z^*$.
\end{prop}

\begin{proof}

Note that if $y_0 \theta_0/x_0 > z^*$ then the optimal strategy includes a sale of the endowed asset at time zero, and the effect of the sale is to move to new state variables $(X^*_0, y_0, \Theta^*_0, 0)$ with the property that $Z^*_0 = y_0 \Theta^*_0/X^*_0 = z^*$. Thus we may assume that $Z_0 = y_0 \Theta^*_0/X^*_0 \leq z^*$.

Consider the equation
\begin{equation}
\label{eq:JL}
 \hat{J}_t = \hat{J}_0 - \int_0^t \tilde{\Lambda}(J_s) ds - \int_0^t
\tilde{\Sigma}(J_s) dB^{1}_s -
\int_0^t \tilde{\Gamma}(J_s) dB^{2}_s
+ \hat{L}_t,
\end{equation}
subject to $\hat{J}_0 = (z^* - z_0)^+$. This equation is associated with a stochastic differential equation with reflection (Revuz and Yor~\cite[p385]{RY}) and has a unique solution $(J, L)$ for which $(J, L)$ is adapted, $J \geq 0$, and $L$ is increasing and continuous, $L_0  =0$ and $L$ only increases when $J$ is zero. Let $(J,L)$ be the solution to (\ref{eq:JL}) with these properties.

Note that $\tilde{\Lambda}(z^*) = \Lambda(0) = 0 = \tilde{\Gamma}(z^*) = \tilde{\Sigma}(z^*)$ and hence $J$ is bounded above by $z^*$.

Now let $Z^*_t = z^* - J_t$, and $\Theta^*_t = \Theta^*_0 \exp \{ - L_t/(z^*(1+z^*)) \}$, and note that whenever $L$ is increasing, (equivalently $\Theta$ is decreasing) we have $Z = z^*$. It follows that the dynamics of $Z$ are governed by
\[
dZ^*_t =  \Lambda(Z^*_t) dt + \Sigma(Z^*_t)dB^1_t + \Gamma(Z^*_t) dB^2_t + Z^*_t(1 + Z^*_t) \frac{d\Theta^*_t}{\Theta^*_t}.
\]
Now set
$X^*_t = Y_t \Theta^*_t/Z^*_t$, $C^*_t = X^*_t[g(Z^*_t) - Z^*_t g'(Z^*_t) /(1-R)]^{-1/R}$ and
$\Pi^*_t = \frac{\lambda}{\sigma} X^*_t \Psi_g(Z^*_t)$.
Then $X^*$ and $C^*$ are positive and adapted and moreover
\begin{eqnarray*}
dX^*_t & = & \frac{Y_t\Theta^*_t}{Z^*_t} \left[\frac{d\Theta^*_t}{\Theta^*_t} + \frac{dY_t}{Y_t} - \frac{dZ^*_t}{Z^*_t} +\left(\frac{dZ^*_t}{Z^*_t}\right)^2 - \frac{dY_t}{Y_t}\frac{dZ^*_t}{Z^*_t}  \right]\\
& = & X_t \bigg\{\left[\alpha - \frac{\Lambda(Z_t^*)}{Z_t^*} + \frac{\sigma (Z_t^*)^2}{(Z_t^*)^2} + \frac{\Gamma(Z_t^*)^2}{(Z_t^*)^2} + 2\rho \frac{\Gamma(Z_t^*) \Sigma(Z_t^*)}{(Z_t^*)^2} -\eta \rho \frac{\Sigma(Z_t^*)}{Z_t^*} - \eta \frac{\Gamma(Z_t^*)}{Z_t^*}\right]dt  \\
& & \hspace{10mm} + \left(\eta - \frac{\Gamma(Z_t^*)}{Z_t^*}\right)dB_t^2 - \frac{\Sigma(Z_t^*)}{Z_t^*} dB_t^1 - Z_t^* \frac{d\Theta^*_t}{\Theta^*_t} \bigg\} \\
& = & \left(\lambda \sigma \Pi^*_t + rX^*_t - C^*_t\right) dt + \sigma \Pi^*_tdB_t^1 - Y_t d\Theta^*_t,
\end{eqnarray*}
where we use the definitions of $\Lambda$, $\Sigma$, $\Gamma$ and $\Psi_g$ for the final equality. It follows immediately that $X^*$ is the wealth process arising from the consumption, portfolio and sale strategy $(C^*, \Pi^*, \Theta^*)$.
\end{proof}

\begin{proof}[Proof of Theorem \ref{thm:maincase}.]
First we show that there is a strategy such that the
candidate value function is
attained, and hence that $V \geq G$.

Observe first that if $z_0 = y_0 \theta_0/x_0 > z^*$ then
\[ \theta_0 - \Theta_0^* = \theta_0 \left( 1 - \frac{z^*}{1+z^*}
\frac{1+ z_0}{z_0} \right) \]
and
\[ X^*_0 = x_0 + y_0(\theta_0 - \Theta^*_0) = x_0 \frac{(1+
z_0)}{(1+z^*)} \]
so that $y \Theta_0^*/X^*_0 = z^*$. Then, since $g(z^*)/g(z_0) = (1+z^*)^{1-R}/(1+z_0)^{1-R}$, for $z_0>z^*$ we have
\[ G(X_0^*,y_0,\Theta_0^*,0) =
\frac{(X^*_0)^{1-R}}{1-R} g(z^*) =
\frac{x_0^{1-R}}{1-R} g(z_0) = \sG(x_0,y_0,\theta_0) .
\]
Hence, without loss of generality we may assume that $z_0 \leq z^*$ since if $z_0>z^*$ the agent transacts from $(x_0,y_0,\theta_0)$ to $(X^*_0, y_0, \Theta^*_0)$ at time zero with no change in value function.

For a general admissible strategy define the
process $M=\left(M_{t}\right)_{t\geq0}$ by
\begin{equation}
M_{t}=\int_{0}^{t}e^{-\beta
s}\frac{C_{s}^{1-R}}{1-R}ds+
e^{-\beta t} \sG\left(X_{t},Y_{t},\Theta_{t}\right).\label{eq:15}
\end{equation}
Write $M^*$ for the corresponding process under the proposed
optimal strategy.
Then $M^*_0 = \sG(X_0^*,y_0,\Theta_0^*) = \sG(x_0,y_0,\theta_0)$ so there is no jump
of $M^*$ at $t=0$.
Further, although the optimal strategy may include the sale of a
positive quantity of the risky asset at time zero, it follows from
Proposition~\ref{prop:Z} that
thereafter the process $\Theta^*$ is continuous and such that
$Z^*_t = Y_t
\Theta^*_t / X^*_t \leq z^*$.

From the form of the candidate value function and the definition of
$g$ given in
(\ref{eq:884}), we know that $\sG$
is $\mathbb{C}^{2,2,1}$.
Then applying It\^{o}'s formula to $M_{t}$, using the continuity of $X^*$ and
$\Theta^*$ for $t>0$, and writing
$\sG_\cdot$
as  shorthand
for $\sG_\cdot(X^*_s,Y_s, \Theta^*_s)$
we have
\begin{eqnarray}
M^*_t - M_0 & = & \int_0^t e^{-\beta s}\Big[ \frac{(C^*_s)^{1-R}}{1-R} -C^*_s \sG_{x}+\alpha Y_s\sG_{y} +\lambda\sigma\Pi^*_s \sG_{x}+rX^*_s \sG_{x} \nonumber \\
& & \hspace{15mm} +\frac{1}{2}\sigma^{2}{\Pi^*_s}^{2}\sG_{xx}+\frac{1}{2}\eta^{2}Y_s^{2}\sG_{yy}+\sigma\eta\rho Y_s \Pi^*_s \sG_{xy} - \beta \sG \Big]ds \nonumber \\
& & + \int_0^t e^{-\beta s} (\sG_\theta - Y_s \sG_x) d\Theta^*_s \nonumber \\
& & +  \int_0^t e^{-\beta s} \sigma \Pi^*_s \sG_x dB^1_s \nonumber \\
& & + \int_0^t e^{-\beta s} \eta Y_s \sG_y dB^2_s \nonumber \\
& = & N_t^1 + N_t^2 + N_t^4 + N_t^5. \nonumber
\end{eqnarray}
Since $Z^*_t \leq z^*$, and since $C^*_t =  \sG_x^{-1/R}$ and $\sL \sG - \beta \sG= 0$ for $z \leq z^*$ we
have $N_{t}^{1}=0$. Further, $d \Theta_s \neq 0$ if and only if
$Z^*_t = z^*$ and then $\sM \sG = 0$, so that $N_{t}^{2}=0$.

To complete the proof of the theorem we need the following lemma proved in Appendix~\ref{app:mgs}.

\begin{lem}
\begin{enumerate}
\item $N^4$ given by $N^4_t = \int_0^t e^{-\beta s} \sigma \Pi^*_s \sG_x dB^1_s$ is a martingale under the optimal strategy.
\item $N^5$ given by $N^5_t =  \int_0^t e^{-\beta s} \eta Y_s \sG_y dB^2_s$ is a martingale under the optimal strategy.
\item $\lim_{t \to \infty} \mathbb{E}[e^{-\beta t} \sG(X^*_t, Y_t, \Theta^*_t)] = 0$.
\end{enumerate}
\label{lem:N4N5}
\end{lem}

Returning to the proof of the theorem, and taking expectations on both sides of $M_t ^* - M_0$, we have
\(
\mathbb{E}\left[M^*_{t}\right]=M_{0},
\)
which leads to
\begin{equation}
G\left(x_{0},y_{0},\theta_{0},0\right)
=\mathbb{E}\left(\int_{0}^{t}e^{-\beta
s}\frac{(C^*_{s})^{1-R}}{1-R}ds\right)
+\mathbb{E}\left[e^{-\beta t}\sG\left(X^*_t,Y_t,\Theta^*_t\right)\right].\label{eq:13}
\end{equation}
Then using Lemma~\ref{lem:N4N5} and applying monotone convergence theorem,
we have
\[
G\left(x_{0},y_{0},\theta_{0},0\right)
= \mathbb{E}\left(\int_{0}^{\infty}e^{-\beta
s}\frac{C_{s}^{*1-R}}{1-R}ds\right) \]
and hence $V \geq G$.

Now we consider general admissible strategies and show that $V \leq G$. Exactly as in (\ref{eq:Mexpansion}) and the proof of Theorem~\ref{thm:deg1} we have
\[ M_t - M_0 = N_t^1 + N_t^2 + N_t^3 + N_t^4 + N_t^5. \]

Lemma~\ref{lem:operator} implies that under general admissible
strategies, $N_{t}^{1}\leq0,\; N_{t}^{2}\leq0$.
Consider the jump term,
\[ 
N_{t}^{3}= \sum_{0<s\leq t}e^{-\beta s}
\left[\sG\left(X_{s},Y_{s},\Theta_{s}\right)
-\sG\left(X_{s-},Y_{s},\Theta_{s-}\right) - \sG_x (\Delta X)_s
 - \sG_\theta (\Delta \Theta)_s
\right]
\] 
Using the fact that $(\Delta X)_s = -Y_s (\Delta \Theta)_s$ and
writing $\theta = \Theta_{s-}$, $y = Y_s$, $x = X_{s-}$, $\chi = - (\Delta
\Theta)_s$ each non-zero jump in $N^3$ is of the form
\[
(\Delta N^3)_s = e^{-\beta s} \Big\{ \sG(x+y \chi,y,\theta - \chi) - \sG(x,y,\theta,s) +
\chi \left[ \sG_\theta (x,y,\theta,s) - y \sG_x (x,y,\theta,s) \right] \Big\}.
\]
Note that by Lemma~\ref{lem:concave}, $\sG(x+y \chi,y,\theta - \chi)$ is
concave in $\chi$ and hence
$(\Delta N^3) \leq 0$.

For $R<1$ the rest of the proof is exactly as in Theorem \ref{thm:deg1}. The case of $R>1$ will be proved in Appendix~\ref{app:R>1}.
\end{proof}

\section{The Verification Lemma in the second non-degenerate case (scenario 3) with no finite critical exercise ratio.}
Throughout this section we suppose that $b_3 \geq b_{3,crit}(b_1, b_2, R)$ and $b_3 < \frac{b_1}{1-R} + b_2 R$
if $R<1$. It follows that $q^* = 1$ and $z^* = \infty$.

Recall the definition of $\gamma$ in (\ref{eqn:gammadefR<1}) or (\ref{eqn:gammadefR>1}), set $h=\gamma^{-1}$ and let $g$ be given by $g(z) = (R b_4/b_1)^R h(\ln z)$. Define the candidate value function as $G(x,y,\theta) = e^{-\beta t} \sG(x,y,\theta)$ where
\begin{equation}
\sG(x,y,\theta)  = \frac{x^{1-R}}{1-R} g\left(\frac{y\theta}{x}\right), \hspace{10mm} x>0, y >0, \theta \geq 0.
\label{eq:vf2}
\end{equation}
We extend the definition to $y \theta < x \leq 0$ via
\[ g(x,y,\theta) = \frac{(x+y \theta)^{1-R}}{1-R} \left(\frac{Rb_4}{b_1}\right)^R n(1)^{-R}. \]

\begin{proof}[Proof of Theorem~\ref{thm:ndeg2}]
Consider first the following stochastic differential equation with reflection
\[ K_t =  K_0 + \int_0^t \hat{\Lambda}(K_s) ds +\int_0^t \hat{\Sigma}(K_s) dB^{1}_s
+ \int_0^t \hat{\Gamma}(K_s) dB^{2}_s
+ L_t, \]
for which $K_0 = x_0 /(y_0 \theta_0)$. By the same argument as in Proposition~\ref{prop:Z}, this equation has a unique solution $(K,L)$ which is an adapted continuous process for which $K$ is non-negative, $L_0=0$ and $L$ only increases when $K$ is zero.

Let $\Theta^*_t = \theta_0 e^{-L_t}$, $X^*_t = Y_t \Theta^*_t K_t$, $C^*_t = X_t^* [g(1/K_t) - \frac{g'(1/K_t)}{K_t(1-R)}]^{-1/R}$ and $\Pi^*_t = \frac{\lambda}{\sigma} X^*_t \Psi_g(1/K_t)$.
Then $\Theta^*$ is decreasing and $X^* \geq 0$. Then with
\(
dK_t =  \hat{\Lambda}(K_t) dt + \hat{\Sigma}(K_t)dB^1_t + \hat{\Gamma}(K_t) dB^2_t -  \frac{d\Theta^*_t}{\Theta^*_t},
\) and, using $K_t dL_t = 0$ and hence $K_t d \Theta^*_t=0$ also,
\begin{eqnarray*}
dX^*_t & = & d(Y_t \Theta^*_t K_t) =  \Theta^*_t K_t dY_t + Y_t \Theta^*_t dK_t + \Theta^*_t d[Y,K]_t \\
& = & X^*_t\left\{\left[\alpha + \frac{\hat{\Lambda}(K_t)}{K_t} + \frac{\eta \hat{\Gamma}(K_t)}{K_t} + \frac{\eta \rho \hat{\Sigma}(K_t)}{K_t}\right]dt + \frac{\hat{\Sigma}(K_t)}{K_t} dB^1_t + \left(\eta + \frac{\hat{\Gamma}(K_t)}{K_t}\right) dB_t^2 \right\} - Y_t d\Theta^*_t \\
& = & \left(\lambda \sigma \Pi^*_t + rX^*_t - C^*_t\right) dt + \sigma \Pi^*_tdB_t^1 - Y_t d\Theta^*_t
\end{eqnarray*}
where we use the definitions of $\hat{\Lambda}$, $\hat{\Gamma}$ and $\hat{\Sigma}$ for the final equality.
It follows immediately that $X^*$ is the wealth process arising from the consumption and sale strategy $(C^*, \Pi^*, \Theta^*)$, and hence that $X^*$ is admissible.

Given admissibility of the candidate optimal strategy, the rest of the proof follows exactly as in the proof of Theorem~\ref{thm:maincase}, except that Lemmas~\ref{lem:concave}, \ref{lem:operator} and \ref{lem:N4N5} are replaced by the following three lemmas, the proofs of which follow in an identical fashion.
\begin{lem}
\label{lem:2ndconcave}
Fix $y$. Then for $x \geq 0$, $\sG=\sG(x,\theta)$ is concave in $x$ and
$\theta$. In particular, if $\psi(\chi) = \sG(x - \chi y \phi, y , \theta + \chi \phi)$, then $\psi$ is concave in $\chi$.
\end{lem}

\begin{lem}
\label{lem:2ndoperator}
Consider the candidate value function constructed in (\ref{eq:vf}). Then for $x \geq 0$, $\mathcal{L}\sG - \beta \sG= 0$, and $\mathcal{M} G \geq 0$ with equality at $x=0$.
\end{lem}

\begin{lem}
\begin{enumerate}
\item $N^4$ given by $N^4_t = \int_0^t e^{-\beta s} \sigma \Pi^*_s \sG_x dB^1_s$ is a martingale under the optimal strategy.
\item $N^5$ given by $N^5_t =  \int_0^t e^{-\beta s} \eta Y_s \sG_y dB^2_s$ is a martingale under the optimal strategy.
\item $\lim_{t \to \infty} \mathbb{E}[e^{-\beta t} \sG(X^*_t, Y_t, \Theta^*_t)] = 0$.
\end{enumerate}
\label{lem:N4N5b}
\end{lem}

\end{proof}

\section{Comparative Statics}

The key to our results on comparative statics is contained in Lemma~\ref{lem:monbi} and Corollary~\ref{cor:pmonb}. In the section on the
problem formulation we showed how the original parameters only affect the solution via four key parameters $b_1$, $b_2$, $b_3$ and $b_4$. Here $b_3$ is the effective Sharpe ratio of the endowed asset and measures the excess expected return, net of any expected growth from correlation with the market asset. $b_1$ is an effective discount parameter, taking account of the investment opportunities in the market. $b_4$ is a measure of the idiosyncratic risk of the endowed asset, and only affects the solution via a scaling of the value function --- idiosyncratic risk also enters the other parameters $b_i$. Finally, $b_2$ is the hardest parameter to interpret, but is a measure of the extent to which the investment motive and the hedging motive cancel with each other.

In this discussion we focus on the critical ratio $z^* = q^*/(1-q^*)$ of endowed wealth to liquid wealth at which sales occur, and the certainty equivalent value $p = p(x,y,\theta)$ of the holdings in the endowed asset. From Lemma~\ref{lem:monbi} we conclude that the critical ratio is decreasing in $b_1$ and $b_2$ and increasing in $b_3$ and from Corollary~\ref{cor:pmonb} we conclude that the certainty equivalent value is similarly decreasing in $b_1$ and $b_2$ and increasing in $b_3$.

The monotonicity in $b_3$ is straightforward to interpret, and has a clear intuition. The greater the effective Sharpe ratio the more valuable the holdings
in the endowed asset, and the longer the agent should hold units of it in her portfolio. The dependence on $b_1$ is also as expected. The greater the effective
discount parameter, the greater the incentive to bring forward consumption which needs to be financed by sales of the risky asset --- thus the endowed
asset is sold sooner. Then there is less opportunity for the benefits from the expected growth of the endowed asset to be enjoyed, thus reducing its value.
 The dependence on $b_2$ is less easy to interpret, but the lemma and its corollary show that there is monotonicity in this parameter also.

In the preceding paragraphs we have discussed the comparative statics in terms of the derived parameters. In order to understand the comparative statics with
respect to the original parameters $r$, $\beta$, $\mu$, $\sigma$, $\alpha$, $\eta$ and $\rho$ we need to consider how the auxiliary parameters depend on these original parameters.

The parameters $\beta$ (discount rate) and $\alpha$ (mean return on the endowed asset) only affect one of the parameters $(b_i)_{i=1,2,3}$, and hence the
comparative statics for these parameters are straightforward. In particular, decreasing $\beta$ or increasing $\alpha$ increases the critical ratio of endowed wealth to liquid wealth $z^*$ at which sales occur, and increases the certainty equivalent value $p=p(x,y,\theta)$. However, the parameters $r$, $\mu$, $\sigma$, $\eta$ and $\rho$ each enter into the definitions of
$b_1$, $b_2$ and $b_3$. Hence the comparative statics with respect to these parameters is more complicated, and in general there is no monotonicity of the
critical ratio or the certainty equivalent value with respect to any of these parameters. For example, an increase in the volatility
$\eta$ of the endowed asset decreases
$b_1$ and may increase or decrease the value of $b_2$ or $b_3$ depending on the values of other parameters. Thus, the effects of a change in the volatility of the
endowed asset on the critical ratio or on the certainty equivalent price are generally mixed.

We restrict our comments on the consumption and investment rate to the following observation about the critical ratio and the Merton line.

Consider an investor who
is free to buy and sell units of $Y$ with zero transaction cost. Then solving the classical Merton problem with two risky assets we find that it is optimal for the
agent to invest a constant fraction  $\frac{(\zeta - \lambda \rho)}{R \eta (1-\rho^2)} = \frac{b_3}{2R}$ of their total wealth in the risky asset.

In contrast, the constrained investor chooses to keep the fraction of his total wealth invested in the endowed asset below $q^*$, i.e. to choose $\Theta_t$ to
ensure that $\frac{Y_t \Theta_t}{X_t} \leq z^* = \frac{q^*}{1-q^*} =z^*$, or equivalently
$\frac{Y_t \Theta_t}{X_t + Y_t \Theta_t} \leq q^*$.
But, it follows from Lemma~\ref{lem:n} that $q^* > \frac{b_3}{2R}$. Hence the `no-sale' region for the constrained investor contains
in its interior the `Merton line' of portfolio
positions for the unconstrained investor.

\appendix
\section{Properties of $n$: proofs}
\label{app:n}

\begin{proof}[Proof of Lemma~\ref{lem:varphi}]
The values of $\varphi$ and $\upsilon$  at $m(q)$ and $\ell(q)$ follow on substitution.
If $b_2 > 1$ then it is immediate from the definition that $(1-R)\upsilon(q,n)<0$. To see that $\upsilon$ is increasing in $n$, note that
\[
\frac{\partial \upsilon}{\partial n} = \frac{\partial \varphi}{\partial n}\left\{1 - \frac{\sgn(1-R) \varphi}{\sqrt{\varphi^2 + E(q)^2}}\right\} >0,
\]
where we use that $\partial \varphi/\partial n = b_1 > 0$. Finally to see that $(1-R)\upsilon$ is concave in $n$, note that
\[
\sgn(1-R)\frac{\partial^2 \upsilon}{\partial n^2} = - \frac{b_1^2 E(q)^2}
{\left(\varphi^2 + E(q)^2 \right)^{3/2}}<0.
\]
\end{proof}

\begin{proof}[Proof of Lemma~\ref{lem:n'}]
First we prove the equivalence of (\ref{eq:nA}) and (\ref{eq:nB}).

Consider
\begin{eqnarray*}
\lefteqn{b_1(\ell(q) - n) + \varphi(q, n)} \\
 & = & R(1-R)q^2 - b_3(1-R)q + b_1 - b_1 n + (1-R)q(1-q) +
\frac{(b_2 - 1)R(1-R)q}{(1-R)q + R}\\
& & \hspace{5mm} + b_1 n - b_1 + b_3(1-R)q + R(1-R)[-2q + 2 - b_2]\\
& = & R(1-R)\left[(1 - q)^2 - (b_2 - 1) + \frac{(b_2 - 1)q}{(1 - R)q + R}\right] + (1-R)q(1-q) \\
& = & (1-R)(1-q)[R(1-q) + q] - \frac{(b_2 - 1)R^2(1-R)}{(1-R)q + R} (1-q).
\end{eqnarray*}
Then, noting that $(1-R)q + R= R(1-q) + q$,
\[
b_1 [(1-R)q + R](\ell(q) - n) = (1-R)(1-q)[R(1-q) + q]^2 - R^2(1-R)(b_2 -  1)(1-q) - \varphi(q,n)[R(1-q) + q],
\]
and multiplying by $4(1-R)(1-q)$,
\begin{eqnarray*}
\lefteqn{4b_1(1-R)(1-q)[(1-R)q + R](\ell(q) - n)} \\
& = & 4(1-R)^2(1-q)^2[R(1-q) + q]^2 - 4\varphi(q,n)(1-R)(1-q)[R(1-q) + q] + \varphi(q,n)^2 \\
& & \hspace{5mm} - \{\sgn(1-R)\}^2 \left(\varphi(q,n)^2 + 4R^2(1-R)^2(b_2 - 1)(1-q)^2\right) \\
& = & \left\{2(1-R)(1-q)[R(1-q)+q] - \varphi(q,n)\right\}^2 -\left\{\sgn(1-R)\right\}^2
\left\{\varphi(q,n)^2 + E(q)^2\right\}.
\end{eqnarray*}
Writing this last expression as the difference of two squares we find
\begin{eqnarray*}
\lefteqn{
\frac{4b_1(1-R) (1-q) [(1-R)q + R](\ell(q) - n)}{2(1-R)(1-q)[(1-R)q + R]
- \varphi(q,n) -\sgn(1-R)\sqrt{\varphi(q,n)^2 + E(q)^2}}
} \\
&= & 2(1-R)(1-q)[R(1-q) + q] - \upsilon(q,n),
\end{eqnarray*}
from which the result follows, on dividing by
$2b_1 R(1-q)[(1-R)q + R] (\ell(q) - n)/((1-R)q)$.

Now consider the equivalence of (\ref{eq:nA}) and (\ref{eq:nC}). We have, starting with (\ref{eq:nA}),
\begin{eqnarray*}
\lefteqn{\frac{(1-R)n}{R(1-q)} \left\{1 - \frac{(1-R)q(1-q)}{b_1 (\ell(q) - n)} + \frac{q\upsilon(q,n)}
{2b_1 [(1-R)q + R] (\ell(q) - n)}\right\} } \\
& = & \frac{(1-R)n \left\{2b_1 (\ell(q) - n) [(1-R)q + R]
-2[(1-R)q + R](1-R)q(1-q) + q\upsilon(q,n) \right\}}{2b_1R[(1-R)q + R](1-q) (\ell(q) - n)} \\
& = & \frac{(1-R)n \left\{2b_1 [(1-R)q + R] \left[(\ell(q) - m(q)) - (n - m(q)) - \frac{(1-R)q(1-q)}{b_1}\right] + q\upsilon(q,n)\right\}}
{2R(1-q)[(1-R)q + R]\left\{S(q) - b_1(n - m(q))\right\}}.
\end{eqnarray*}
The result then follows since
\[
2b_1 [(1-R)q + R] \left\{\ell(q) - m(q) - \frac{(1-R)q(1-q)}{b_1}\right\} = 2R(1-R)(b_2 - 1)q = -q\upsilon(q,m).
\]
\end{proof}

\begin{proof}[Proof of Lemma~\ref{lem:n}]

(1) From the expression (\ref{eq:n}) and l'H\^{o}pital's rule,
$n'(0)=\chi$ solves
\[
\chi = \frac{1-R}{R} - \frac{(1-R)^2}{b_1 R} \frac{1}{\ell'(0) - \chi}+
\frac{(1-R)}{2 b_1 R^2} \frac{\upsilon(0,1)}{\ell'(0) - \chi} ,
\]
where we have
\[
\upsilon(0,1) = 2R(1-R) - b_2 R(1-R) - \sgn(1-R) b_2 R |1-R| = 2R(1-R)(1 - b_2),
\]
and $\ell'(0) = (b_2 - b_3)(1-R)/b_1$. This gives
\[
\chi = \frac{1-R}{R} - \frac{b_2 (1-R)^2}{R\left[(b_2 - b_3) (1-R) - b_1 \chi \right]},
\]
or equivalently, we have that $\chi$ solves $\Phi(\chi) = 0$.
Further,
\[
\Phi\left(\ell'(0)\right) = \Phi\left(\frac{(b_2 - b_3)(1-R)}{b_1}\right) = -(1-R)^2  b_2 < 0.
\]
For $R<1$, we have $n'(0)< \ell'(0)$ by hypothesis, so that $n'(0)$ is the smaller root
of $\Phi$. For $R>1$, we have $n'(0)>\ell'(0)$ by hypothesis and $n'(0)$ is the larger root of $\Phi$.

(2) For $R<1$, $n'(0)<\ell'(0)$ so that initially $n<\ell$. Then, from (\ref{eq:nC}), $\lim_{n \uparrow \ell(q)} F(q,n) =- \infty$. Hence $n(q)< \ell(q)$, at least until $q=1$ or $\ell$ hits zero.
The argument for $R>1$ is similar.

(3) It is clear from (\ref{eq:nC}) that $F(q,m(q))=0$. Also, for $q \leq q_n$ so that $(1-R)(\ell - n) > 0$,
the sign of $F(q,n(q))$ is opposite to the sign of the factor $D = D(q, m(q), n(q))$ where
\[D(q,m(q),n) = 2 b_1 [(1-R)q+R](n-m(q)) - qv(q, n) + qv(q,m(q)).\]
But $\partial \varphi/\partial n = b_1$, and so
\[ \frac{\partial D}{\partial n} = 2 b_1 [(1-R)q+R] - q b_1 \left[ 1 - \frac{\sgn(1-R) \varphi(q,n)}{ \sqrt{\varphi(q,n)^2 + E(q)^2}} \right] > 2 b_1 [(1-R)q+R]
- 2 q b_1 = 2R(1-q) b_1 > 0. \]
Hence, $D$ is increasing in $n$ and $D(q,m(q),n)>0$ if and only if $n(q)>m(q)$.

(4) If $R>1$ then $n(q)$ is increasing on $[0,q^*]$. In particular, $q_n > q^*$ unless $q^*=1$ whence $q_n = 1 = q^*$. If $R<1$, then $n(q) \leq \ell(q)$ on $(0,1)$. But from (\ref{eq:nA}) we see that $n$ cannot hit zero strictly before $\ell$. The result follows since $\ell$ is concave, so $q_{\ell} < 1$ if and only if $\ell(1)<0$.

(5) We can only have $q^*<1$ if $(1-R)m'(1)>0$.
For $R<1$, we must have $n'(q^*) = 0 < m'(q^*)$. But $m$ has a minimum at
$b_3/2R$, so $q^*> b_3/2R$. For $R>1$, we must
have $n'(q^*) = 0 > m'(q^*)$. But $m$ has a maximum at
$b_3/2R$, so $q^*> b_3/2R$.

\end{proof}

\begin{proof}[Proof of Lemma~\ref{lem:monbi}]
We consider monotonicity in $b_3$. The monotonicity results for $b_2$ and $b_1$ can be proved in a similar fashion, and are generally easier, since in the case of $b_2$, $m$ does not depend on $b_2$ and in the case of $b_1$ the dependence of $m$ on $b_1$ can be eliminated by a change of co-ordinates to $\hat{m}(q) = b_1(m(q)-1)$, with similar expressions for $\hat{n}$ and $\hat{l}$.

For fixed $b_1 > 0$, $b_2 \geq 1$ and $R$, write $n(\cdot) = n(\cdot; b_3)$. Consider $\phi(b_3) = n'(0;b_3)$. Differentiating $\Phi$ with respect to $b_3$ we find
\[
\frac{\partial \phi}{\partial b_3}  =  \frac{(1-R)\left(\frac{1-R}{R} - \phi \right)}{2b_1 \phi + (1-R)\left(b_3 - b_2 - \frac{b_1}{R}\right)}
 =  - \frac{(1-R)}{b_1} \left\{ \frac{\left(\frac{1-R}{R} - \phi \right)}{\left(\frac{1-R}{R} - \phi \right) + b_1\left(\frac{(1-R)(b_2 - b_3)}{b_1} - \phi \right)} \right\}.
\]
Suppose $R<1$. Then $\phi<0$ and $\phi<\ell'(0) = (1-R)(b_2 - b_3)/b_1$, and
\begin{equation}
0>\frac{\partial \phi}{\partial b_3} > -\frac{1-R}{b_1} = \frac{\partial m'(0)}{\partial b_3}.
\label{eq:zetainequal}
\end{equation}
Now suppose $R>1$. Then $\phi>0$ and $\phi>\ell'(0) = (1-R)(b_2 - b_3)/b_1$, and $0<\partial \phi/\partial b_3 < (R - 1)/b_1 = \partial m'(0)/\partial b_3$.

Differentiating (\ref{eq:nB}) with respect to $b_3$ we find
\[
\frac{\partial}{\partial b_3} F(q,n;b_3)\bigg|_{q = \tilde{q}, n = \tilde{n}} = \frac{2(1-R)^2 \tilde{q} \tilde{n}}
{RD(\tilde{q}, \tilde{n}; b_3)^2} \frac{\partial}{\partial b_3} D(\tilde{q}, \tilde{n}; b_3),
\]
where
$D(q, n; b_2) = 2(1-R)(1-q)[(1-R)q + R] - \varphi(q, n; b_2) - \sgn(1-R)\sqrt{\varphi(q,n;b_2)^2 + E(q; b_2)^2}$.
Then
\[
\frac{\partial}{\partial b_3} D(q, n; b_3) = - \frac{1}{\sqrt{\varphi(q,n;b_3)^2 + E(q)^2}}  \left\{\sqrt{\varphi(q,n;b_3)^2 + E(q)^2} + \sgn(1-R)\varphi(q,n;b_3)\right\} \frac{\partial \varphi}{\partial b_3} .
\]
Since $\partial \varphi/\partial b_3 = (1-R)q$, we find $\partial F/\partial b_3|_{q = \tilde{q}, n = \tilde{n}}$ has the opposite sign to $1-R$.

Now suppose $\overline{b}_3>\underline{b}_3$ and let $\overline{n}(q) = n(q; \overline{b}_3)$ and $\underline{n}(q) = n(q; \underline{b}_3)$. Comparing derivatives at zero we conclude that initially $(1-R)\overline{n} < (1-R)\underline{n}$. Also if $\overline{n}(\cdot)$ and $\underline{n}(\cdot)$ cross at some point $(\tilde{q}, \tilde{n})$ with $\tilde{q} < q^*(\overline{b}_3) \wedge q^*(\underline{b}_3)$ then it follows from our knowledge of $\frac{\partial}{\partial b_3} F(q,n; b_3)$ that $(1-R)\overline{n}'<(1-R)\underline{n}'$. But at a first crossing, $(1-R)\overline{n}$ must cross $(1-R)\underline{n}$ from below which is a contradiction. Hence $(1-R)n(\cdot; b_3)$ is decreasing in $b_3$ for $q \leq q^*$.

This alone is not sufficient to make conclusions about $q^*(b_3)$ since $m$ also depends on $b_3$.
Define $\zeta(q;b_3) = n(q;b_3) - m(q;b_3)$, and let $m_0(q) = 1 + R(1-R)q^2/b_1$ so that $m(q;b_3) = m_0(q) - b_3(1-R)q/b_1$. Then $d \zeta/ d q = H(q,\zeta;b_3)$ where
\[
H(q, \zeta; b_3) = - \frac{(1-R)B_1(q,\zeta(q;b_3); b_3) B_2(q,\zeta(q;b_3); b_3)}{2R(1-q)[(1-R)q + R] \left\{S(q) - b_1 \zeta\right\}} - \frac{2R(1-R)q}{b_1} + \frac{b_3(1-R)}{b_1},
\]
and
\begin{eqnarray*}
B_1(q,\zeta;b_3) & = & \zeta + m_0(q) - \frac{b_3(1-R)}{b_1} q \\
B_2(q,\zeta;b_3) & = & 2b_1[(1-R)q + R]\zeta - q\upsilon \left(q, m_0(q) - \frac{b_3(1-R)}{b_1} + \zeta \right) + q\upsilon \left(q, m_0(q) - \frac{b_3(1-R)}{b_1}\right).
\end{eqnarray*}
Suppose $R<1$. Fix $q$ and $\zeta > -m(q)$. Then $B_1$ is positive and decreasing in $b_3$, and $B_2$ is positive and decreasing in $b_3$, since ${\upsilon}$ is concave in $n$. Hence the product $B_1 B_2$ is decreasing in $b_3$ and $H$ is increasing in $b_3$.
By the result in (\ref{eq:zetainequal}), $\zeta'(0) = \phi - m'(0)$ is increasing in $b_3$, so that at least initially, $\zeta$ is increasing in $b_3$. Then since $H$ is increasing in $b_3$, it follows that solutions of $\zeta'(q;b_3) = H(q,\zeta;b_3)$ for different $b_3$ cannot cross, and hence $\zeta(q;b_3)$ is increasing in $b_3$. Thus $q^*(b_3)$ is increasing in $b_3$. Similar arguments apply when $R>1$.

Finally, consider $\frac{\partial}{\partial b_3} n(q^*(b_3))$  over the interval where $q^*(b_3) \in (0,1)$. We have
\[
\frac{\partial}{\partial b_3}n(q^*(b_3)) = \frac{\partial n}{\partial b_3}\bigg|_{q^*(b_3)} + n'(q^*(b_3))\frac{\partial q^*}{\partial b_3}.
\]
But $n'(q^*) = 0$ and hence $n(q^*(b_3))$ is decreasing since $n(q;b_3)$ is decreasing in $b_3$.
\end{proof}

\begin{proof}[Proof of Corollary~\ref{cor:pmonb}]
In the case where $0<q^*<1$ extend the domain of $h$ to $(-\infty,\infty)$ by
\[ h(u) = \frac{(1+e^u)^{1-R}}{(1+ e^{u^*})^{1-R}} h^*  \hspace{10mm}  u > u^*. \]
Then $g(z) = (b_4 R/b_1)^R h(\ln z)$ for all $z \in [0,\infty)$. The monotonicity results for $p$ will follow if $(1-R)h$ is decreasing in $b_1$ and $b_2$ and
increasing in $b_3$.

We focus on monotonicity in $b_1$ for the case $R<1$; the proof of monotonicity in $b_2$ and $b_3$ and for $R>1$ follows similarly.

Fix $b_2$ and $b_3>0$ and suppose $b_1 > ((1-R)b_3 - b_2 R)^+$, to ensure that we are not in the case where the value function is infinite.

Given $n(q;b_1)$ defined on $[0, q^*(b_1)]$ extend the domain of definition to $[0,1]$ by setting $n(q;b_1) = n(q^*(b_1);b_1)$ for $q > q^*(b_1)$.
Let $N(q;b_1) = (1-q)^{-(1-R)} n(q;b_1)^{-R}$ defined on $[0,1]$ and let $W(\cdot;b_1)$ be inverse to $N(\cdot;b_1)$.

Then, for each $q \in (0,1]$, $n(q;b_1)$ is decreasing in $b_1$ and $N(q;b_1)$ is increasing in $b_1$.
It follows that $W(\cdot;b_1)$ and $w(\cdot;b_1)$ are increasing in $b_1$.

We know that $\lim_{u \uparrow \infty} e^{-(1-R)u}h(u) = n(1)^{-R} = n(q^*)^{-R}$ is decreasing in $b_1$.
We want to argue that $h$ is decreasing in $b_1$ for all $u$.

Fix $\underline{b}_1 < \overline{b}_1$. Suppose there exists $u \in (-\infty,\infty)$ such that $h(u;\underline{b}_1) = h(u;\overline{b}_1)$
and let
$\tilde{u}$ be the largest such $u$; set $h(\tilde{u};\underline{b}_1) = h(\tilde{u};\overline{b}_1)
=\tilde{h}$. Then $h(u;\underline{b}_1) < h(u;\overline{b}_1)$ for all $u > \tilde{u}$, and we must have
$\frac{dh}{du}(\tilde{u};\underline{b}_1) \leq \frac{dh}{du}(\tilde{u};\overline{b}_1)$, or equivalently
$w(\tilde{h};\underline{b}_1) \leq w(\tilde{h};\overline{b}_1)$. But $w$ is increasing in $b_1$ contradicting the hypothesis that $h$ is not
decreasing in $b_1$.

\end{proof}

\begin{proof}[Proof of Lemma~\ref{lem:b3critbound}]
If $b_3>0$ and $m$ is monotonic then we must have $q^*=1$. Hence $b_{3,crit} \leq 2R$.

If $b_3>0$, $R<1$ and $\ell(1) \leq 0$ then $q^*=1$. Hence for $R<1$ and $\ell(1) \leq 0$ we must have $b_{3,crit} \leq b_2 R + \frac{b_1}{1-R}$.

These arguments show that $b_{3,crit}(b_1,1,R) \leq \bar{b}_3$. It remains to show that for $b_2=1$ and $0< b_3 < \bar{b}_3$ we have $q^*=1$.

If $b_2=1$, then $\ell(1)=m(1)$, and provided $b_3 < \bar{b}_3$, $\ell(1)>0$. Suppose $R<1$, the case $R>1$ being easier. If $b_3 \leq R$ then $m(1) \geq m(0)=1$ and since $n$ is decreasing we must have $q^*<1$.

So suppose $R<b_3 < \bar{b}_3$. By the arguments in Proposition~\ref{prop:NWw}(3) we have $n(q) \leq k(q) = 1 + \frac{q(R-b_3)(1-R)}{b_1}$. Also
\[ (1-q) F(q,m(1)) = m(1) \left[ \frac{1-R}{R} - \frac{(1-R)^2}{b_1 R}  q \frac{(1-q)}{\ell(q) - \ell(1)} \right] \rightarrow
m(1) \frac{1-R}{R} \left[ 1 + \frac{(1-R)}{b_1}   \frac{1}{\ell'(1)} \right]. \]
But $\frac{b_1 \ell'(1)}{1-R} = 2R - b_3 - 1 = (R-1) + (R-b_3) \in (-1,0)$. Hence $\lim_{q \uparrow 1} (1-q)F(q,m(1)) = \kappa \in (-\infty,0)$. Since for $n \geq m(1)$ we must have $n$ crosses the horizontal line at height $m(1)$ before $q=1$. Then also $q^*<1$.

(2)
Note first that
\[\lim_{b_2 \to \infty} \frac{\ell(q)}{b_2} = \frac{R(1-R)}{b_1} \frac{q}{(1-R)q + R},
\hspace{7mm}
\lim_{b_2 \to \infty} \frac{\upsilon(q,n)}{b_2} = -2R(1-R).
\]
It follows immediately by l'H\^{o}pital's rule that
\begin{equation}
\lim_{b_2 \to \infty} F(q,n)  =  n \left\{ \frac{1-R}{R(1-q)} + \frac{(1-R)q}{2b_1R(1-q)[(1-R)q + R]} \lim_{b_2 \to \infty} \frac{v(q,n)/b_2}{\ell(q)/b_2} \right\} = 0.
\end{equation}
Then, if $n_\infty(q) = \lim_{b_2 \uparrow \infty}n(q;b_2)$ we have $n_{\infty}'(q)=0$,
which implies $n(q) = n(0) = 1$. It is easy to see that $n_\infty$ crosses $m$ at some $q^* \in (0,1)$ if and only if $(1-R)m(1) > (1-R)m(0) = (1-R)$ which is equivalent to $0<b_3<R$. Otherwise, we have $q^* = 1$.

The final statement follows from the monotonicity of $q^*$ with respect to $b_2$ and $b_3$.

\end{proof}

\begin{proof}[Proof of Proposition~\ref{prop:NWw}.]

(1) $N$ solves
\[
\frac{N'(q)}{N(q)} = \frac{2(1-R)^2q(1-q)[(1-R)q + R] - (1-R)q \upsilon\left(q, (1-q)^{1 - \frac{1}{R}}N(q)^{-\frac{1}{R}}\right)}{2b_1 (1-q) [(1-R)q + R] \left\{\ell(q) - (1-q)^{1 - \frac{1}{R}}N(q)^{-\frac{1}{R}}\right\}}
\]
Suppose $R<1$. We have $\upsilon(q,n) \leq 0$ (see Lemma~\ref{lem:varphi})
and, for $q \leq q^*$, $n(q)=(1 - q)^{1 - 1/R} N(q)^{-1/R} \leq \ell(q)$. Hence $N$ is increasing. For $R>1$, $\upsilon (q, (1-q)^{1 - 1/R}N(q)) \geq 0$ and $n(q)=(1 - q)^{1 - 1/R} N(q)^{-1/R} \geq \ell(q)$. Hence, $N$ is decreasing.

(2) Note that if $s=N(q)$ then $W'(s)=1/N'(q)$ and so
$W$ solves
\[ 
W'(s) = \frac{2b_1 (1-W(s)) [(1-R)W(s) + R] \left\{\ell(W(s)) - (1-W(s))^{1 - \frac{1}{R}}s^{-\frac{1}{R}}\right\}}{s\left\{2(1-R)^2 W(s) (1-W(s))[(1-R)W(s) + R] - (1-R)W(s) \upsilon\left(W(s), (1-W(s))^{1 - \frac{1}{R}}s^{-\frac{1}{R}}\right)\right\}}
\] 
The expression for $w'$ follows after some lengthy algebra.

(3) We have that $n'(q) = F(q, n(q))$ and then from the representation (\ref{eq:nB}) and the definition of $N$ it follows that
\[
\frac{N'(q)}{N(q)} =  \frac{2(1-R)^2 q}{2(1-R)^2(1-q)\left(q + \frac{R}{1-R}\right)-\varphi\left(q,n(q)\right)
  -\sgn(1-R) \sqrt{\varphi\left(q,n(q)\right)^2 + E(q)^2}},
\]
where $n(q) = (1-q)^{1 - \frac{1}{R}} N(q)^{-1/R}$.
Then an alternative representation for $W'$ is
\[
s(1-R)W'(s) =  \frac{2(1-R)(1-W(s))\left((1-R)W(s) + R\right) -\hat{\varphi}
 - \sgn(1-R) \sqrt{\hat{\varphi}^2 + E(W(s))^2}}{2(1-R)W(s)}.
\]
where $\hat{\varphi} = \varphi\left(W(s),(1-W(s))^{1 - \frac{1}{R}} s\right)$.

We want to show for $s \in (1, h^*)$, $1 - R < w'(s) < 1 - \frac{Rw(s)}{(1-R)s}$. This is equivalent to
\begin{equation}
\label{eq:Wineq}
\frac{(1-R)(1 - W(s))}{s}  <  (1-R)W'(s)  <  \frac{1 - W(s)}{s} .
\end{equation}
Consider the second of these inequalities. It can be rewritten as either
$\frac{1}{1-R} \frac{N'(q)}{N(q)}  >  \frac{1}{1-q}$ or $(1-R) \frac{n'(q)}{n(q)}  <  0$, the latter of which is immediate since $(1-R)n$ is decreasing on $(0,q^*)$. Now consider the first inequality in (\ref{eq:Wineq}). If $R>1$ then it is immediate since $W$ is decreasing.
It only remains to show that $n'/n>-1/(1-q)$ for $R<1$.

Suppose $R<1$ and let $k(q)$ solve $\frac{F(q,k)}{k} = -\frac{1}{1-q}$. (Note that $F(q,m(q))=0 $, $F(q, \ell(q)) = - \infty$ and
$\varphi(q,k) + \sqrt{\varphi(q,k)^2 + E(q)^2}$ is increasing in $k$ so that from (\ref{eq:nB}), $F(q,k)/k$ is decreasing in $k$.
In particular, $\frac{F(q,k)}{k} = -\frac{1}{1-q}$ has a unique solution in $(m(q),\ell(q))$.) Then $k$ is the straight line with $k(0) = 1$ and $k(1) = \ell(1)$, so that
\begin{equation}
\label{eq:kdef}
 k(q) = 1 + q \frac{(b_2 R - b_3)(1-R)}{b_1} .
\end{equation}
To verify this, note that if $k$ is as given in (\ref{eq:kdef}) then $\varphi(q,k(q)) = R(1-R)(2-b_2)(1-q)$ and $\upsilon(q,k(q))= - 2 R(1-R)(b_2-1)(1-q)$.
The desired conclusion
$F(q,k(q)) = -k(q)/(1-q)$ follows after some lengthy algebra.

We have that $\Phi(k'(0)) = R(1-R)^3 b_2\left\{b_3 - \left(b_2 R + \frac{b_1}{1-R}\right)\right\}/b_1$.
By assumption $b_3 < b_2 R+ b_1/(1-R)$ (else we are in a degenerate case), and then $\Phi(k'(0))<0$ and $n'(0)<k'(0)$, so that $n<k$ on some interval $(0,q)$
to the right of zero.
Suppose now that $n(q) = k(q)$ for some $q \in (0,q^*)$. Then
\[
n'(q) = F(q,n(q)) = -\frac{k(q)}{1-q} = - \frac{1 + qk'(q)}{1-q} < k'(q).
\]
Let $\tilde{q}$ be the smallest such point, then $\tilde{q}$ is a downcrossing of $n$ over $k$, contradicting the fact that $n(q)<k(q)$ on $(0,\tilde{q})$.
Hence $n(q) < k(q)$ for any $q \in (0, q^*)$. But for $n \leq k(q)$, $F(q, n) > - \frac{n}{1-q}$. Thus $\frac{n'}{n} > - \frac{1}{1-q}$ and $w'(s)> 1-R$.

\end{proof}

\section{The martingale property of the value function}
\label{app:mgs}

\begin{proof}[Proof of Lemma~\ref{lem:N4N5} and Lemma~\ref{lem:N4N5b}.]
(i) As a first step of the proof, we show that $\sigma \Pi^* \sG_x/\sG$ is bounded. From the form of the candidate value function (\ref{eq:vf}) and the optimal portfolio process in Theorem~\ref{thm:maincase} we have
\begin{eqnarray}
\lefteqn{\sigma \Pi^*(x,y,\theta) \frac{\sG_x(x,y,\theta)}{\sG(x,y,\theta)} } \nonumber \\
& = & \lambda \left[(1-R) - \frac{zg'(z)}{g(z)}\right] \left\{ \frac{-(1-R)g(z) + \left(1 + \frac{R\eta \rho}{\lambda}\right)zg'(z) + \frac{\eta \rho}{\lambda} z^2g''(z)}{-R(1-R)g(z) + 2Rzg'(z) + z^2g''(z) } \right\} \nonumber \\
& = & \frac{(1-R)}{R}\left(1 - \frac{w(h)}{(1-R)h}\right)
\left( \lambda  - \frac{\left(\eta \rho- \frac{\lambda}{R}\right)w(h)[w'(h) - (1 - R)]}{(1-R)h - \left(2 - \frac{1}{R}\right) w(h) - \frac{1}{R}w(h)w'(h)} \right) \nonumber \\
& = & \frac{(1-R)}{R}\left(1 - W(h)\right)
\left( \lambda - \frac{\left(\eta \rho - \frac{\lambda}{R}\right)W(h) [w'(h) - (1 - R)]}
{1 - \left(2 - \frac{1}{R}\right)W(h) - \frac{1-R}{R} W(h)^2 - \frac{1-R}{R} hW(h)W^{'}(h)}  \right) \nonumber
\end{eqnarray}
Now use the fact that $(1-R) < w'(h) < 1-RW(h)$ to conclude that $0 < w'(h)-(1-R) < R(1-W(h))$ and
$h(1-R)W'(h)<1-W(h)$ to conclude that
\begin{eqnarray*}
\lefteqn{1  - \left(2 - \frac{1}{R}\right)W(h) - \frac{1-R}{R} W(h)^2 - \frac{1-R}{R} hW(h)W^{'}(h)} \\
& \geq & 1  - \left(2 - \frac{1}{R}\right)W(h) - \frac{1-R}{R} W(h)^2 - \frac{1}{R} W(h)(1 - W(h)) \\
& = & 1 - 2W(h) +W(h)^2.
\end{eqnarray*}
Then
\begin{equation}
\left| \sigma \Pi^* \frac{\sG_x}{\sG} \right|
\leq  \frac{|1-R|}{R}\left(1 - W(h)\right)
\left( |\lambda| +  \frac{|\eta \rho R - \lambda| W(h)}{1-W(h)}  \right) \leq \frac{|1-R|}{R} (|\lambda| + |\eta \rho R - \lambda|) =: K_\pi.
\label{eq:pibd}
\end{equation}
Note that this bound applies for both the proof of Lemma \ref{lem:N4N5} and Lemma~\ref{lem:N4N5}.

Now we want to show that
the local martingale
\[
N_{t}^{5}=\int_{0}^{t}e^{-\beta s} \eta Y_s \sG_{y}(X^*_s, Y_s, \Theta^*_s) dB^2_{s}
\]
is a martingale. This will follow if, for example,
\begin{equation}
\mathbb{E}\int_{0}^{t}e^{-2\beta s}\left(Y_s \sG_{y}(X^*_s, Y_s, \Theta^*_s)
\right)^{2}ds<\infty\label{eq:4}
\end{equation}
for each $t>0$.
From the form of the value function (\ref{eq:vf}), we have
\begin{equation}
\label{eqn:Gy}
e^{-\beta s} y G_y(x,y,\theta)
= e^{-\beta s}\frac{x^{1-R}}{1-R}zg'\left(z\right)
= e^{-\beta s} \sG(x,y,\theta,t)
\frac{zg'\left(z\right)}{g\left(z\right)}
\leq (1-R) e^{-\beta s} \sG\left(x,y,\theta\right)
\end{equation}
where we use that
$\frac{zg'\left(z\right)}{g\left(z\right)}
=\frac{w\left(h\right)}{h} = (1-R)W(h)$ and $0 \leq W(h) \leq 1$.

Define a process $\left(D_{t}\right)_{t\geq0}$ by
$ 
D_{t}=\ln G\left(X^*_{t},Y_{t},\Theta^*_{t},t\right) = \ln \sG(X^*_{t},Y_{t},\Theta^*_{t}) - \beta t
$. 
Then $D$ solves
\begin{eqnarray*}
D_t - D_0
& = & -\int_{0}^{t}\frac{1}{1-R}\frac{1}{\sG}\sG_{x}^{\frac{R-1}{R}}ds
+\int_0^t \frac{1}{\sG} \sigma \Pi^*_s \sG_x dB^1_s + \int_0^t \frac{1}{\sG} \eta Y_s \sG_y dB^2_s \\
& & \hspace{10mm}- \frac{1}{2} \int_0^t \frac{1}{\sG^2} \left[\sigma^2 {\Pi^*_s}^2 \sG_x^2 + \eta^2 Y_s^2 \sG_y^2 + 2\sigma \eta \rho Y_s \Pi^*_s \sG_x \sG_y \right]ds.
\end{eqnarray*}
It follows that the candidate value function along the optimal trajectory has the
representation
\begin{equation}
\label{eqn:G0}
G\left(X^*_t,Y_t,\Theta^*_t,t\right)=
G\left(X^*_{0},y_{0},\Theta^*_{0},0\right) \exp\left\{
- \int_{0}^{t}\frac{1}{1-R}\frac{1}{\sG}\sG_{x}^{\frac{R-1}{R}}ds\right\} H_t
\end{equation}
where $H=\left(H_{t}\right)_{t\geq0}$ is the exponential martingale
\[ H_{t}  =  \sE \left( \frac{\sigma \Pi^* \sG_x}{\sG} \circ B^1 + \frac{\eta Y_s \sG_y}{\sG} \circ B^2 \right)_t , \]
where $\sE( A \circ B)_t = e^{\int_0^t A_s dB_s - \frac{1}{2} \int_0^t A^2_s ds}$.
Note that (\ref{eqn:Gy}) implies
$0 \leq \frac{1}{\sG}\eta y\sG_{y} \leq\eta |1-R|$ and (\ref{eq:pibd}) implies that $|\frac{1}{\sG}\sigma \Pi^* \sG_x |\leq K_\pi$,
so that $H$ is indeed a martingale, and not merely a local martingale.

From (\ref{eqn:Gy}) and (\ref{eqn:G0}), we have
\begin{eqnarray*}
\left(yG_{y}\right)^{2}
& = &
G\left(X_{0},y_{0},\Theta_{0},0\right)^{2} \left(\frac{zg^{'}
\left(z\right)}{g\left(z\right)}\right)^{2}\exp\left\{
-2\int_{0}^{t}\frac{1}{1-R}\frac{1}{\sG}\sG_{x}^{\frac{R-1}{R}}ds\right\} H_t^2
\\
& \leq  &
G\left(X_{0},y_{0},\Theta_{0},0\right)^{2} (1-R)^2
H_t^2.
\end{eqnarray*}
But from (\ref{eq:pibd}) and (\ref{eqn:Gy})
\begin{eqnarray*}
H_t^2 & = &
\sE \left( \frac{2}{\sG}\sigma \Pi^* \sG_x \circ B^1 + \frac{2}{\sG}\eta Y_s \sG_y \circ B^2 \right)_t
\exp\left\{\int_0^t \frac{1}{\sG^2} \left[\sigma^2 ({\Pi^*_s})^2 \sG_x^2 + \eta^2 Y_s^2 \sG_y^2 + 2\sigma \eta \rho Y_s \Pi^*_s \sG_x \sG_y \right]ds\right\} \\
& \leq & \sE \left( \frac{2}{\sG}\sigma \Pi^* \sG_x \circ B^1 + \frac{2}{\sG}\eta Y_s \sG_y \circ B^2 \right)_t
\exp\left\{[\eta^2 (1-R)^2 + K_{\pi}^2 + 2 \eta |\rho 1-R)| K_{\pi}]t\right\}
\end{eqnarray*}
Hence $\E[H_t^2] \leq \exp\left\{[\eta^2 (1-R)^2 + K_{\pi}^2 + 2 \eta |\rho (1-R)| K_{\pi}]t\right\}$ and it follows that
(\ref{eq:4}) holds
for every $t$, and hence that the local martingale
$N_{t}^{5}=\int_{0}^{t}\eta yV_{y}dB_{s}$ is a martingale
under the optimal strategy.

(ii) The same reasoning applies to show that the local martingale
\[
N_{t}^{4}=\int_{0}^{t} e^{-\beta s} \sigma \Pi^*_s \sG_{x}(X^*_s, Y_s, \Theta^*_s)
dB^1_{s}
\]
is a martingale. We have
$|\sigma \Pi^* G_{x}/G| \leq K_{\pi}$ and by (\ref{eqn:G0}) $(\sigma \Pi^* G_x(x,y,\theta,t))^2 \leq G(X_0, y_0, \Theta_0, 0)^2 K_{\pi}^2 H_t^2$.

(iii) Consider
$\exp\left\{ \int_{0}^{t}-\frac{1}{1-R}\frac{1}{\sG}\sG_{x}^{\frac{R-1}{R}}ds\right\} $.
To date we have merely argued that this function is decreasing in
$t$. Now we want to argue that it decreases to zero exponentially
quickly. By (\ref{eq:vf}), we have
\begin{eqnarray*}
\frac{1}{1-R}\frac{1}{\sG}\sG_{x}^{\frac{R-1}{R}} & = &
\frac{1}{g(z)}\left[g\left(z\right)-\frac{1}{1-R}zg^{'}\left(z\right)\right]^{\frac{R-1}{R}}
  =  \frac{1}{h}\left[h-\frac{1}{1-R}w\left(h\right)\right]^{\frac{R-1}{R}} \\
 & = & h^{-1/R}(1-W(h))^{1-1/R} = N(q)^{-1/R}(1-q)^{1-1/R} = n(q).
\end{eqnarray*}
Then since $n$ is bounded below on $(0,q^*)$ (by 1 if $R>1$ and by $n(q^*)$ if $R<1$) we have a lower bound on all the expressions in the above equation.

Hence from (\ref{eqn:G0}) we have
\[
0 \leq (1-R) G(X^*_t, Y_t, \Theta^*_t, t) \leq (1-R)G(x_0, y_0, \theta_0, 0)e^{- \min \{n(q^*),1\} t} H_t \to 0
\]
and then $G \to 0$ in $L^1$ as required.

\end{proof}

\section{extension to $R>1$}
\label{app:R>1}
It remains to extend the proofs of the verification lemmas to the case $R>1$. In particular we need to show that the candidate value function is an upper bound on the value function. The main idea is taken from Davis and Norman\cite{Davis}.

Suppose $G\left(x,y,\theta,t\right) = e^{-\beta t} \sG(x,y,\theta)$ is the candidate value function. Consider for $\varepsilon>0$,
\begin{equation}
\widetilde{V}_\varepsilon \left(x,y,\theta,t\right) = \widetilde{V}\left(x,y,\theta,t\right) =G\left(x+\varepsilon,y,\theta,t\right)\label{eq:2ndvr}
\end{equation}
and suppose $\widetilde{M}_t = \widetilde{M}_t (C, \Theta)$ is given by
 \[
 \widetilde{M}_t = \int_0^t e^{-\beta s}\frac{C_s ^{1-R}}{1-R} ds + \widetilde{V} (X_t, Y_t, \Theta_t, t).
 \]
Then,
\begin{eqnarray*}
\widetilde{M}_t - \widetilde{M}_0 & = & \int_0^t \Big[ e^{-\beta s}U\left(C_s\right)-C_s \widetilde{V}_{x}+\alpha Y_s \widetilde{V}_{y}+ \widetilde{V}_{s}+\lambda\sigma\Pi_s \widetilde{V}_{x}+rX_s \widetilde{V}_{x} \\
& & \hspace{10mm}+\frac{1}{2}\sigma^{2}\Pi_s^{2} \widetilde{V}_{xx}+\frac{1}{2}\eta^{2}Y_s^{2} \widetilde{V}_{yy}+\sigma\eta\rho Y_s\Pi_s \widetilde{V}_{xy} \Big]ds \\
& & + \int_0^t (\widetilde{V}_\theta - Y_s \widetilde{V}_x) d\Theta_s \\
& & + \sum_{\substack{0<s \leq t}}
\left[\widetilde{V}(X_s, Y_s, \Theta_s, s) - \widetilde{V}(X_{s-}, Y_{s-}, \Theta_{s-}, s-) - \widetilde{V}_x (\bigtriangleup X)_s
- \widetilde{V}_\theta (\bigtriangleup \Theta)_s \right] \\
& & +  \int_0^t \sigma \Pi_s \widetilde{V}_x dB^1_s \\
& & + \int_0^t \eta Y_s \widetilde{V}_y dB^2_s \\
& = & N_t^1 + N_t^2 + N_t^3 + N_t^4 + N_t^5.
\end{eqnarray*}

Lemma~\ref{deg1:inequal} (in the case $b_3 \leq 0$ and otherwise Lemma~\ref{lem:operator} or Lemma~\ref{lem:2ndoperator}) implies $\widetilde{N}_{t}^{1} \leq 0$ and $\widetilde{N}_{t}^{2} \leq 0$. The concavity of $\widetilde{V}(x + y\chi, y, \theta - \chi)$ in $\chi$ (either directly if $b_3 \leq 0$, or using Lemma~\ref{lem:concave} and Lemma~\ref{lem:2ndconcave}) implies $(\Delta \widetilde{N}^3) \leq 0$.

Now define stopping times
$\tau^1_{n}=\inf\left\{ t\geq0: \int_0^t \sigma^2 \Pi_s^2 \widetilde{V}_x ^2 ds\geq n\right\}$ and
$\tau^2_{n}=\inf\left\{ t\geq0:\int_{0}^{t}\eta^{2}Y_s^{2}\widetilde{V}_{y}^{2}ds\geq n\right\}$. Let $\tau_{n} = \min\left\{\tau^1_n, \tau^2_n\right\}$ and $\tau_{n}$ is a stopping time.
It follows from (\ref{eq:pibd}) and (\ref{eqn:Gy}) that $\Pi_t \widetilde{V}_x$ and $y\tilde{V}_y$ are bounded and hence $\tau^1_n\to \infty$, $\tau^2_n\to \infty$ and hence $\tau_n \to \infty$. Then the local martingale $(\widetilde{N}_{t \wedge \tau_n}^4 + \widetilde{N}_{t \wedge \tau_n}^5)_{t \geq 0}$ is a martingale and taking expectations we have $\mathbb{E}\left(\widetilde{M}_{t\wedge\tau_{n}}\right)\leq\widetilde{M}_{0}$, and hence
\[
\mathbb{E}\left(\int_{0}^{t\wedge\tau_{n}}
e^{-\beta s}\frac{C_{s}^{1-R}}{1-R}ds+\widetilde{V}\left(X_{t\wedge\tau_{n}},Y_{t\wedge\tau_{n}},\theta_{t\wedge\tau_{n}},t\wedge\tau_{n}\right)\right)\leq\widetilde{V}\left(x_{0},y_{0},\theta_{0},0\right).
\]

In the case $b_3 \leq 0$, (\ref{eq:vdeg1}) and (\ref{eq:2ndvr}) imply
\begin{eqnarray*}
\widetilde{V}\left(X_{t\wedge\tau_{n}},Y_{t\wedge\tau_{n}},\theta_{t\wedge\tau_{n}},t\wedge\tau_{n}\right) & = & e^{-\beta (t\wedge\tau_{n})}\frac{\left(X_{t\wedge\tau_{n}}+
\varepsilon \right)^{1-R}}{1-R}
\left(1 + \frac{Y_{t\wedge\tau_{n}}\theta_{t\wedge\tau_{n}}}
{X_{t\wedge\tau_{n}}+\varepsilon}\right)^{1-R} \left(\frac{b_1}{b_4 R}\right)^{-R}\\
& \geq & e^{-\beta (t\wedge\tau_{n})}\frac{\left(X_{t\wedge\tau_{n}}+
\varepsilon \right)^{1-R}}{1-R}
 \left(\frac{b_1}{b_4 R}\right)^{-R} \geq \frac{\varepsilon^{1-R}}{1-R} \left(\frac{b_1}{b_4 R}\right)^{-R}.
\end{eqnarray*}
Thus $\widetilde{V}$ is bounded, $\lim_{n\to \infty} \mathbb{E}\widetilde{V}\left(X_{t\wedge\tau_{n}},Y_{t\wedge\tau_{n}},\theta_{t\wedge\tau_{n}},t\wedge\tau_{n}\right) = \mathbb{E}\left[\widetilde{V}\left(X_{t},Y_{t},\theta_{t},t\right)\right]$, and
\[
\widetilde{V}\left(x_{0},y_{0},\theta_{0},0\right)\geq\mathbb{E}\left(\int_{0}^{t}e^{-\beta s}\frac{C_{s}^{1-R}}{1-R}ds\right)+\mathbb{E}\left[\widetilde{V}\left(X_t, Y_t,\Theta_t, t\right)\right].
\]
Similarly,
\[
\widetilde{V}(x,y,\theta,t) \geq e^{-\beta t} \frac{\varepsilon^{1-R}}{1-R} \left(\frac{b_1}{b_4 R}\right)^{-R},
\]
and hence $\mathbb{E}\left[\widetilde{V}(X_t, Y_t, \Theta_t, t)\right] \to 0$.
Then letting $t \to \infty$ and applying the monotone convergence theorem, we have
\[
\widetilde{V}_{\varepsilon} (x_0, y_0, \theta_0, 0) = \widetilde{V}_{\varepsilon} (x_0, y_0, \theta_0, 0) \geq \mathbb{E}\left(\int_0^\infty e^{-\beta s} \frac{C_s ^{1-R}}{1-R} ds\right).
\]
Finally let $\varepsilon \to 0$. Then $V \leq \lim_{\varepsilon \to 0} \widetilde{V} = G$. Hence, we have $V \leq G$.

The two non-degenerate cases are very similar, except that now from (\ref{eq:vf}) or (\ref{eq:vf2}) and (\ref{eq:2ndvr}),
\[
\widetilde{V}(x,y,\theta) = \frac{(x + \varepsilon)^{1-R}}{1-R} g\left(\frac{y\theta}{x + \varepsilon}\right) \geq \frac{\varepsilon^{1-R}}{1-R}\left(\frac{b_1}{b_4 R}\right)^{-R},
\]
where we use that for $R>1$, $g$ is decreasing with
$g(0) = \left(\frac{b_1}{b_4 R}\right)^{-R}>0$. Hence $\widetilde{V}$ is bounded, and the argument proceeds as before.

\end{document}